\documentclass{article}
\usepackage[utf8]{inputenc}
\usepackage{dsfont}
\usepackage{mathtools}
\usepackage{geometry}
\usepackage{amsmath}
\usepackage{amsthm}
\usepackage{amssymb}
\usepackage{xcolor}
\usepackage{comment}
\usepackage{tensor}
\usepackage{graphicx}
\usepackage{caption}
\usepackage{subcaption}
\usepackage[normalem]{ulem}
\numberwithin{equation}{section}
\usepackage{hyperref}
\newtheorem{hypothesis}{Hypothesis}
\newtheorem{theorem}{Theorem}[section]
\newtheorem{theoremmain}{Theorem}
\newtheorem{theoremmainagain}{Theorem}
\newtheorem*{theorem*}{Theorem}
\newtheorem{example}{Example}[section]
\newtheorem{proposition}{Proposition}[section]
\newtheorem*{proposition*}{Proposition}
\newtheorem{remark}{Remark}[section]
\newtheorem{lemma}{Lemma}[section]
\newtheorem*{lemma*}{Lemma}
\newtheorem{definition}{Definition}[section]
\newtheorem*{proposal*}{Proposal}
\newtheorem{corollary}{Corollary}[section]
\newtheorem*{corollary*}{Corollary}

\title{Detecting screens modeled by Schr\"odinger operators that generate $C_0$ contraction semigroups }
\author{Lawrence Frolov\thanks{Department of Mathematics, Rutgers University (New Brunswick)}
}
\date{October 2025}

\begin{document}
\maketitle
\begin{abstract}
    Consider a non-relativistic quantum particle with wave function $\psi$ in a bounded $C^2$ region $\Omega \subset \mathbb{R}^n$, and suppose detectors are placed along the boundary $\partial \Omega$. Assume the detection process is irreversible, its mechanism is time independent and also hard, i.e., detections occur only along the boundary $\partial \Omega$. Under these conditions Tumulka \cite{TUMULKA2022168910} informally argued that the dynamics of $\psi$ must be governed by a $C_0$ contraction semigroup that weakly solves the Schr\"odinger equation and proposed modeling the detector by a time-independent local absorbing boundary condition at $\partial \Omega$. In this paper, we apply the newly discovered theory of boundary quadruples \cite{Arendt2023} to parameterize all $C_0$ contraction semigroups whose generators extend the Schr\"odinger Hamiltonian, and prove a variant of Tumulka's claim: all such evolutions are generated by the placement of a linear absorbing boundary condition on $\psi$ along $\partial \Omega$. We combine this result with the work of Werner \cite{Werner1987} to show that each $C_0$ contraction semigroup naturally admits a Born rule for the time of detection along $\partial \Omega$, and we prove that a detection will almost surely occur in finite time if detectors have been placed everywhere along $\partial \Omega$.
\end{abstract}
\section{Introduction and Statement of Main Results}

\subsection{Irreversible Hard Autonomous Detection}
Suppose that a detecting surface, such as a scintillating screen, is placed along the boundary $\partial \Omega$ of a bounded region $\Omega \subset \mathbb{R}^3$ in physical space, and suppose that a non-relativistic quantum particle is prepared at time $0$ with wave function $\psi_0$ whose support lies in $\Omega$. We would like to understand how placing this detector along the boundary $\partial \Omega$ affects the dynamics of the wave function $\psi$ in $\Omega$. We expect that as detections occur, the total probability that the particle remains undetected in $\Omega$ decreases. As an idealization, we take the detector to be \textit{hard}, so detections only occur along the boundary $\partial \Omega$. The detection process should also be \textit{irreversible}, i.e. probability in the particle-detector state space should irreversibly transfer from the space of undetected states to the space of detected states. We additionally assume that the mechanism of detection is time-independent, it depends only on the initially prepared state of the detector. This paper aims to provide an explicit parameterization of all quantum-mechanical models which satisfy these assumptions, and to demonstrate that every such model admits a natural Born rule for the distribution of times at which the particle is detected along $\partial \Omega$.
\par
We are primarily motivated by a recent paper of Tumulka \cite{TUMULKA2022168910}, in which he informally argued that the dynamics of a particle interacting with an idealized detecting screen should be governed by a $C_0$ contraction semigroup that weakly solves the Schr\"odinger equation. In summary, Tumulka analyzes the particle and detector together as a quantum system with wave function $\Psi_t$. This wave function evolves unitarily in a Hilbert space of the form $\mathcal{H}_P\oplus\mathcal{H}_F$, where $\mathcal{H}_F$ denotes the space of states in which a detector has fired while $\mathcal{H}_P=L^2(\Omega)\otimes\mathcal{H}_D$ consists of states in which the particle resides in $\Omega$ and the detectors are primed. The system is initially prepared as a pure product state $\Psi_0=\psi_0\otimes\phi_0\in L^2(\Omega)\otimes\mathcal{H}_D$, with $\psi_0$ and $\phi_0$ of unit norm. Tumulka models detection as an irreversible process that transports probability from $\mathcal{H}_P$ to $\mathcal{H}_F$. This means the dynamics of the projection $\Psi_t \big{|}_{\mathcal{H}_P}$ are norm non-increasing and autonomous, they are not affected by the dynamics in $\mathcal{H}_F$.
\par
Since the detector is assumed to be \textit{hard}, detections cannot occur in the interior of $\Omega$. Tumulka also neglects the formation of entanglement in the dynamics of $\Psi_t \big{|}_{\mathcal{H}_P}$, as the only interactions between the particle and the idealized detector are those that lead to immediate firings. So, the wave function of the particle conditioned on non-detection is given by a pure state $\psi_t$ that undergoes a non-relativistic Schr\"odinger evolution in $\Omega$. This returns our first condition (C1).
\begin{enumerate}
    \item[(C1)] $\psi_t$ satisfies a Schr\"odinger equation inside $\Omega$ in a weak distributional sense.
    \begin{equation}\label{Intro Evo Eq}
        i \frac{\partial \psi}{\partial t}= \hat{H}^*\psi \quad \text{in }\Omega
    \end{equation}
    where, setting $\hbar=\frac{\hbar^2}{2m}=1$, $\hat{H}^*$ denotes the adjoint of the Hamiltonian $\hat{H}=-\Delta+V$ defined on $D(\hat{H})=C_c^\infty(\Omega)$ with $V$ a real valued potential depending on the experimental apparatus. 
\end{enumerate}
Time independence of the detection mechanism implies that the evolution of $\psi_t$ should be autonomous, although it may depend on the initial detector state $\phi_0$. Then the dynamics of $\psi_t$ inherits linearity, continuity in time, and the semigroup property from the dynamics of $\Psi_t\big{|}_{\mathcal{H}_P}$.
\begin{enumerate}
    \item[(C2)] The evolution maps $W_t:\psi_0 \mapsto\psi_t$ define a $C_0$ semigroup on $L^2(\Omega)$:
    \begin{enumerate}
        \item The maps $W_t:L^2(\Omega)\to L^2(\Omega)$ are linear for $t \geq 0$.
        \item They are strongly continuous,   $\lim_{t \to t_0}||W_t \psi - W_{t_0}\psi||_{L^2(\Omega)}=0$ for all $\psi \in L^2(\Omega)$, $t_0 \geq 0$.
        \item They form a semigroup under composition, $W_t W_s= W_{t+s}$ for $t,s \geq 0$, with $W_0=\mathds{1}$. 
    \end{enumerate}
\end{enumerate}
The quantity $||\psi_t||_{L^2(\Omega)}^2$ represents the probability that the particle has remained undetected in $\Omega$ up to time $t$, so we expect the dynamics of $\psi_t$ to be non-unitary. This returns the final condition.
\begin{enumerate}
    \item[(C3)] $W_t$ are contractions, $||W_t \psi||_{L^2(\Omega)} \leq ||\psi||_{L^2(\Omega)}$ for all $\psi \in L^2(\Omega)$.
\end{enumerate}
The reasoning above is a slight refinement of the argument presented in \cite{TUMULKA2022168910}, which only implicitly stipulated (C2b) and demanded further conditions that we will show are effectively redundant. It would be interesting to investigate the circumstances in which we expect a real-life detecting screen to be approximately described by such an idealized model, but we do not do that here. Our goal is to study the mathematical implications of conditions (C1-C3).

$C_0$ contraction semigroups have a long history in the study of quantum particles undergoing irreversible interactions with their environment (see e.g. \cite{Davies,Exner}). Allcock \cite{ALLCOCK} famously modeled \textit{soft} detection with the non-unitary evolution operators $W_t=\exp(-it(-\Delta -iV)):L^2(\mathbb{R}^n)\to L^2(\mathbb{R}^n)$, where $iV$ is an imaginary potential. In this model, quantum particles are gradually detected with rate $2V$ upon entering the region where $iV$ is supported. It is tempting to model hard detection along the boundary of a region $\Omega$ using an imaginary potential supported in the complement $\mathbb{R}^n\setminus  \Omega$ and taking a limit of the dynamics as $V \to \infty$. However, this limit famously returns norm-preserving dynamics in $\Omega$, so with probability $1$ the particle is never detected.

Tumulka has instead proposed that the interaction between the particle and the hard detector should be modeled by a local time-independent linear absorbing boundary condition, so that the dynamics of $\psi_t$ in $\Omega$ are governed by the initial-boundary value problem
\begin{equation}\label{IBVP}
\left\{\begin{array}{rclr}
       i \partial_t \psi&=& (-\Delta + V) \psi \quad &\text{in } \Omega
       \\
       \psi &=&\psi_0 \quad &\text{at } t=0
       \\
       \partial_n \psi&=&i \beta \psi \quad &\text{on } \partial \Omega
    \end{array}
    \right.
    \end{equation}
Here, $\partial_n$ denotes the outwards normal derivative of $\Omega$, and $\beta$ is a function on $\partial \Omega$ satisfying $\text{Re}(\beta)\geq 0$. The real part of $\beta$ should be taken strictly positive wherever detectors are present along $\partial \Omega$.

Tumulka's absorbing boundary condition was largely motivated by the question of \textit{detection time distributions}. Quantum mechanics accurately predicts the formation of interference patterns in the distribution of particle positions measured along a screen, but a theoretical description of the distribution of \textit{times} at which a particle is detected on a screen is still an open question \cite{Mielnik}. Various competing proposals have been made for this distribution (see e.g. \cite{MUGA2000353} for a wide review), but Tumulka's model admits a natural Born rule. For initial wave functions satisfying $||\psi_0||_{L^2(\Omega)}=1$, the \textit{absorbing boundary rule} states that over any time interval $0 \leq t_1 <t_2$, the probability of detecting a particle with wave function $\psi_t$ satisfying equation
(\ref{IBVP}) is
\begin{equation}
    \text{Prob}_{\psi_0}(t_1 \leq t \leq t_2)= \int_{t_1}^{t_2} \int_{\partial \Omega} \vec{n}\cdot \vec{j}_{\psi_t}~ dx^{n-1}dt=\int_{t_1}^{t_2}\int_{\partial \Omega} 2\text{Re}(\beta)|\psi_t|^2~ dx^{n-1}dt
\end{equation}
where $dx^{n-1}$ denotes the surface element on $\partial \Omega$, $\vec{n}$ the outwards unit normal and $\vec{j}_{\psi_t}=2\text{Im} (\psi_t^* \vec{\nabla}\psi_t)$ is the probability current. It has been shown in \cite{Existence} that for sufficiently regular $\beta$, the initial-boundary value problem (\ref{IBVP}) admits a unique global-in-time solution $\psi_t$ for each $\psi_0 \in H^2(\Omega)$ satisfying $\partial_n\psi_0 \big{|}_{\partial \Omega}=i\beta\psi_0 \big{|}_{\partial \Omega}$. The solution mappings $W_t:\psi_0 \mapsto\psi_t$ were shown to continuously extend to a unique $C_0$ contraction semigroup on $L^2(\Omega)$, and the detection time probability distributions are well defined for any initial data $\psi_0 \in L^2(\Omega)$. We emphasize that this proposal provides a probability distribution for the times at which the particle is first \textit{detected} along $\partial \Omega$, not the arrival time of the particle in the absence of detectors.
\par
This paper aims to establish a converse of the result in \cite{Existence}. We shall prove that every $C_0$ contraction semigroup on $L^2(\Omega)$ whose generator is extended by $-i\hat{H}^*$ corresponds to the placement of a time-independent linear absorbing boundary condition. In one space dimension this result is known (see e.g. \cite{Arendt2023}), and we will find it instructive to state for the simplest case $\Omega=(-\infty,0)$.
\begin{theorem*}
    Let $\hat{H}:=(-\partial_x^2 + V)\big{|}_{C_c^\infty((-\infty,0))}$ with $V \in L^\infty((-\infty,0],\mathbb{R})$. Then a $C_0$ contraction semigroup $W_t$ on $L^2((-\infty,0])$ has its generator extended by $-i\hat{H}^*$ if and only if there exists a $\Phi \in \mathbb{C}$ with $|\Phi|\leq 1$ such that $W_t$ is the solution mapping of the initial-boundary value problem \em
    \begin{equation}
        \left\{\begin{array}{rclr}
       i \partial_t \psi&=& \hat{H}^* \psi \quad &\text{in } \Omega
       \\
       \psi &=&\psi_0 \quad &\text{at } t=0
       \\
        \psi+i\partial_x\psi&=&\Phi(\psi-i\partial_x \psi) \quad &\text{at } x=0
    \end{array}
    \right.
    \end{equation}
\end{theorem*}
Analogous results in higher space dimensions require considerable effort to state and prove, as it is not immediately clear how to formalize the notion of boundary values for distributional wave functions residing in $D(\hat{H}^*)$. Generalizations of the theorem stated above do exist in the literature in terms of multivalued linear operators (see e.g. \cite{Derkach,BHS2020}); however, we shall offer an alternative characterization using explicit boundary conditions. This characterization generalizes the work of \cite{Facchi}, and naturally lends itself to the statement of a Born rule for the distribution of times at which the particle is detected along $\partial \Omega$.
\par
Our work is based on the theory of boundary tuples, a framework that goes back to the work of von Neumann \cite{NeumannOne,NeumannTwo} and later J.W. Calkin \cite{Calkin} on the parameterization of self-adjoint extensions of symmetric operators. Boundary triples and their modern incarnations have proven robust in parameterizing extensions of symmetric operators, and are notably applied to study spectral properties of elliptic differential operators (see e.g. \cite{BEHRNDTLanger2007,BEHRNDT2014}). Recently, Wegner \cite{Wegner2017}
demonstrated that a large class of $C_0$ contraction semigroups could be parameterized in the boundary triple framework without the use of multivalued linear operators and his result was soon extended by Arendt et al. \cite{Arendt2023} through the newly discovered theory of boundary quadruples.
\par
For a densely defined symmetric operator $\hat{H}$ on a Hilbert space $\mathcal{H}$, a \textit{boundary quadruple} for $-i\hat{H}$ consists of two Hilbert spaces $\mathcal{H}_\pm$ and two linear mappings $G_\pm:D(\hat{H}^*)\to\mathcal{H}_\pm$ such that
\begin{equation}\label{abstract green}
    \langle-i\hat{H}^*\psi, \phi \rangle_{\mathcal{H}} + \langle\psi,-i\hat{H}^*\phi\rangle_{\mathcal{H}}=\langle G_+\psi,G_+\phi\rangle_{\mathcal{H}_+}-\langle G_-\psi,G_- \phi \rangle_{\mathcal{H}_-}
\end{equation}
holds for all $\psi, \phi \in D(\hat{H}^*)$, and the map $(G_+,G_-):D(\hat{H}^*)\to \mathcal{H}_+\times \mathcal{H}_-$ is surjective. When $\hat{H}:C_c^\infty(\Omega)\to L^2(\Omega)$ is a symmetric differential operator on $\mathcal{H}=L^2(\Omega)$, it is often possible to construct a boundary quadruple for $-i\hat{H}$ where $G_\pm \psi\in L^2(\partial \Omega)$ are functions of $\psi$ or its derivatives restricted to $\partial \Omega$, in which case we call equation (\ref{abstract green}) an ``abstract Green's identity". All densely defined symmetric operators $\hat{H}$ admit boundary quadruples, and any quadruple can be used to parameterize the set of $C_0$ contraction semigroups whose generators are extended by $-i\hat{H}^*$.
\begin{theorem*}
    \cite[Theorem 3.10]{Arendt2023}  Let $-i\hat{H}$ be a densely defined skew-symmetric operator on a Hilbert space $\mathcal{H}$, and let $(\mathcal{H}_\pm,G_\pm)$ be a boundary quadruple for $-i\hat{H}$. Then the following are equivalent
    \begin{enumerate}
        \item[(a)] $W_t=\exp(-itL)$ is a $C_0$ contraction semigroup on $\mathcal{H}$ with generator $L\subset \hat{H}^*$.
        \item[(b)] There exists a linear contraction $\Phi:\mathcal{H_-}\to \mathcal{H}_+$ such that $W_t$ is the solution mapping of the abstract initial-boundary value problem  
    \begin{equation}\label{intro abstract IBVP}
    \left\{\begin{array}{rclr}
       i\partial_t \psi&=& \hat{H}^* \psi \quad 
       \\
       \psi\big{|}_{t=0} &=&\psi_0 
       \\
    G_+\psi&=&\Phi G_- \psi
    \end{array}
    \right.
\end{equation}
Stated formally, the generator of $W_t=\exp(-itL)$ takes the form $L=\hat{H}_{\Phi}$ where
\begin{equation}\label{Dom B}
            D(\hat{H}_\Phi):=\{\psi \in D(\hat{H}^*):G_+\psi=\Phi G_-\psi \}, \quad \hat{H}_\Phi\psi:=\hat{H}^*\psi. 
        \end{equation}
    \end{enumerate}
\end{theorem*}
If $\exp(-it\hat{H}_\Phi)$ is a $C_0$ contraction semigroup with $\hat{H}_\Phi$ of the form in equation (\ref{Dom B}), we call $G_+\psi=\Phi G_-\Psi$ an ``abstract absorbing boundary condition" because the norm of $\psi_t=\exp(-it\hat{H}_\Phi)\psi_0$ is lost at rate
    \begin{equation}
        \frac{d}{dt}||\psi_t||_{\mathcal{H}}^2=||\Phi G_-\psi_t||^2_{\mathcal{H}_+}- ||G_- \psi_t||^2_{\mathcal{H}_-}\leq 0.
    \end{equation}
\subsection{Main Results}
Several results in this paper are based on the following regularity assumptions for $\Omega$ and $V$.
\begin{hypothesis}\label{hypothesis}
    We assume $\Omega\subset \mathbb{R}^n$ is a bounded $C^2$ domain of dimension $n\geq 2$, and we take $\hat{H}:=(-\Delta+V)\big{|}_{C_c^\infty(\Omega)}$ with $V \in L^\infty(\Omega,\mathbb{R})$.
\end{hypothesis}

We shall provide an explicit parameterization of all $C_0$ contraction semigroups that ``weakly solve" the Schr\"odinger equation $i\partial_t\psi=\hat{H}^*\psi$ in $\Omega$, proving that they are the solution mappings of a particular class of initial-boundary value problems. 
\begin{theoremmain}\label{main}
    Assume Hypothesis \ref{hypothesis}. Then there exists linear maps $G_\pm:D(\hat{H}^*)\to L^2(\partial \Omega)$ that define a boundary quadruple for $-i\hat{H}$, where $G_\pm \psi$ are linear functions of $\psi$ and its normal derivative $\partial_n \psi$ along $\partial \Omega$ (see Proposition \ref{Current Decomposition}). Consequently, a $C_0$ contraction semigroup $W_t$ on $L^2(\Omega)$ has its generator extended by $-i\hat{H}^*$ if and only if there exists a linear contraction $\Phi:L^2(\partial\Omega)\to L^2(\partial \Omega)$ such that $W_t$ is the solution mapping of the initial-boundary value problem
    \em \begin{equation}\label{main gen IBVP}
\left\{\begin{array}{rclr}
       i \partial_t \psi&=& \hat{H}^* \psi \quad &\text{in } \Omega
       \\
       \psi &=&\psi_0 \quad &\text{at } t=0
       \\
       G_+\psi&=&\Phi G_-\psi \quad &\text{on } \partial \Omega
    \end{array}
    \right.
    \end{equation}
\end{theoremmain}
Hence, all models for idealized hard detection that satisfy conditions (C1-C3) must come from placing an absorbing boundary condition along $\partial \Omega$, and this boundary condition completely specifies the interaction between the particle and the detecting screen while the detector has not yet fired. Most of these $C_0$ contraction semigroups admit highly nonlocal dynamics. This is hardly surprising, given that our assumptions never ruled out dynamics where probability is instantaneously transported from one part of the boundary to another. For the modeling of hard detectors, only local absorbing boundary conditions are relevant. We shall study a large family of such boundary conditions known as the ``Robin" boundary conditions $\partial_n\psi \big{|}_{\partial \Omega}=i\beta \psi \big{|}_{\partial \Omega}$, extending the well-posedness result of \cite{Existence} to an even wider class of boundary operators $\beta$.
\begin{hypothesis}\label{Robin Operator Conditions}
    Suppose that $\beta=\beta_1+\beta_2+\beta_3$ is a sum of three bounded linear operators $\beta_k:H^{1/2}(\partial \Omega)\to H^{-1/2}(\partial \Omega)$ having ``non-negative real parts", meaning \em $\text{Re}\langle \beta_k \zeta,\zeta \rangle_{H^{-1/2}(\partial \Omega)\times H^{1/2}(\partial \Omega)}\geq 0$ \em for all $\zeta\in H^{1/2}(\partial \Omega)$ and each $k\in\{1,2,3\}$, with the following conditions on $\beta_1$, $\beta_2$, and $\beta_3$:
    \begin{enumerate}
        \em \item \em Compactness: $\beta_1:H^{1/2}(\partial \Omega)\to H^{-1/2}(\partial \Omega)$ is a compact operator.
        \em \item \em Smallness in operator norm: $||\beta_2
        ||_{\mathcal{B}(H^{1/2}(\partial \Omega),H^{-1/2}(\partial \Omega))}< ||\tau_D||^{-2}_{\mathcal{B}(H^{1}( \Omega),H^{1/2}(\partial \Omega))}$, where $\tau_D:H^1(\Omega)\to H^{1/2}(\partial \Omega)$ denotes the boundary value operator mapping $\psi \mapsto \psi \big{|}_{\partial \Omega}$.
        \em \item \em Non-negative imaginary part: \em $\text{Im}\langle \beta_3 \zeta, \zeta\rangle_{H^{-1/2}(\partial \Omega)\times H^{1/2}(\partial \Omega)}\geq 0$ \em for all $\zeta \in H^{1/2}(\partial \Omega)$. 
    \end{enumerate}
\end{hypothesis}
These conditions are quite general, allowing the Robin boundary operator to contain small tangential derivatives and; if $\dim(\Omega)>2$; multiplication by unbounded functions with poles of order up to $1$ (see \cite{Hardy} and Remark \ref{Hardy remark}). The special case where $i\beta$ is self-adjoint, for which detectors are absent and $\partial \Omega$ consists only of reflecting walls, has been studied in \cite{Gesztesy09} using the theory of quadratic forms. These methods may be insufficient to study post-detection states and relativistic generalizations of our model, so we take a more general approach in extending those results to the dissipative setting; we shall explicitly construct the linear contractions $\Phi$ on $L^2(\partial \Omega)$ corresponding to the Robin boundary conditions and prove the following well-posedness theorem.  
\begin{theoremmain}\label{main Robin}
   Assume Hypothesis \ref{hypothesis} and suppose that $\beta$ satisfies the conditions of Hypothesis \ref{Robin Operator Conditions}. Then the initial-boundary value problem 
   \em \begin{equation}\label{Gen Robin IBVP}
\left\{\begin{array}{rclr}
       i \partial_t \psi&=& \hat{H}^* \psi \quad &\text{in } \Omega
       \\
       \psi &=&\psi_0 \quad &\text{at } t=0
       \\
       \partial_n \psi&=&i \beta \psi \quad &\text{on } \partial \Omega
    \end{array}
    \right.
    \end{equation}
    admits a unique global-in-time solution $\psi_t\in C^1([0,\infty),L^2(\Omega))$ for all initial data $\psi_0 \in H^1(\Omega)\cap D(\hat{H}^*)$ satisfying $\partial_n \psi_0 \big{|}_{\partial \Omega}=i \beta \psi_0 \big{|}_{\partial \Omega}$. In addition, the solution mappings $W_t:\psi_0 \mapsto \psi_t$ extend continuously to a $C_0$ contraction semigroup on $L^2(\Omega)$.
\end{theoremmain}
For initial states $\psi_0$ of unit norm, the probability that the particle is never detected along $\partial \Omega$ is $\lim_{t \to \infty}||\psi_t||_{L^2(\Omega)}^2$. A natural question is whether this probability becomes zero whenever the particle has been  completely enveloped by detectors, meaning that the Robin boundary operator $\beta$ has \textit{strictly positive real part} on $\partial \Omega$: $\text{Re}\langle \beta \zeta, \zeta \rangle_{H^{-1/2}(\partial \Omega)\times H^{1/2}(\partial \Omega)}>0$ for all $0\neq \zeta \in H^{1/2}(\partial \Omega)$.
\begin{theoremmain}\label{main Robin stability}
Assume Hypothesis \ref{hypothesis} and suppose that $\beta$ satisfies the conditions of Hypothesis \ref{Robin Operator Conditions} while having strictly positive real part. Then all solutions of the initial-boundary value problem (\ref{Gen Robin IBVP}) asymptotically vanish, $||W_t\psi_0||_{L^2(\Omega)}\xrightarrow{t \to \infty}0$ for all $\psi_0 \in L^2(\Omega)$.
\end{theoremmain}
While studying a much broader class of detector models, Werner \cite{Werner1987} showed that any $C_0$ contraction semigroup $W_t$ on a Hilbert space $\mathcal{H}$ admits a positive operator-valued measure $E(\cdot)$ defined on Lebesgue measureable subsets $I \subset [0,\infty)\cup \{\text{N}\}$ and acting on $\mathcal{H}$ which produces the detection time probability distribution
\begin{equation}\label{main Werner distribution}
    \langle E(I)\psi_0,\psi_0\rangle_{\mathcal{H}}=\text{Prob}_{\psi_0}(t\in I):=\begin{cases}
        ||W_{t_1}\psi_0||_{\mathcal{H}}^2- ||W_{t_2}\psi_0||_{\mathcal{H}}^2 \quad &\text{for } I=(t_1,t_2)\subset[0,\infty)
        \\
        \lim_{t \to \infty}||W_{t}\psi_0||_{\mathcal{H}}^2 \quad &\text{for } I=\{\text{N} \}
    \end{cases}
\end{equation}
Here, $t=\text{N}$ denotes the event that the particle is never detected. The construction of this positive operator-valued measure is in general very abstract. The final result of this paper applies the boundary quadruple framework to explicitly construct $E(\cdot)$ whenever the $C_0$ contraction semigroup ``weakly solves" $i\partial_t\psi=\hat{H}^*\psi$ for some densely defined symmetric operator $\hat{H}$. 
\begin{theoremmain}\label{Born Rule Main}
    Let $-i\hat{H}$ be a densely defined skew-symmetric operator on a Hilbert space $\mathcal{H}$ and let $W_t$ be a $C_0$ contraction semigroup on $\mathcal{H}$ whose generator is extended by $-i\hat{H}^*$. Choosing any boundary quadruple $(\mathcal{H}_\pm, G_\pm)$ for $-i\hat{H}$, we may write $W_t=\exp(-it\hat{H}_\Phi)$ where $\Phi: \mathcal{H}_- \to \mathcal{H}_+$ is the linear contraction in (\ref{Dom B}). Then the linear mapping $\psi_0 \mapsto \sqrt{1-\Phi^*\Phi}G_-W_t\psi_0$ defines a bounded operator $J:\mathcal{H}\to L^2([0,\infty),\mathcal{H}_-)$ such that $||\left(J\psi_0\right) (t)||_{\mathcal{H_-}}^2$ is a (unnormalized) probability density for the detection time distribution \em 
    \begin{equation}
     \text{Prob}_{\psi_0}(t_1 \leq t \leq t_2):=||W_{t_1} \psi_0||_{\mathcal{H}}^2 - ||W_{t_2}\psi_0||_{\mathcal{H}}^2=\int_{t_1}^{t_2}||\left(J\psi_0\right) (t)||_{\mathcal{H}_-}^2 ~dt.
    \end{equation}
    \em Consequently, the positive operator-valued measure $E(\cdot)$ that produces the detection time probability distribution(\ref{main Werner distribution}) can be expressed in terms of the indicator function $\chi_I$ as $E(I)=J^*\chi_IJ$ for Lebesgue measureable subsets $I\subset [0,\infty)$ and $E(\{N\})=\mathds{1}_\mathcal{H}-J^*J$. 
\end{theoremmain}
The organization of this paper is as follows. Section 2 reviews the theory of boundary quadruples. Section 3 reviews the standard boundary tuple construction for Schr\"odinger operators on bounded $C^2$ domains, culminating in a proof of Theorem \ref{main}. Section 4 discusses particular initial-boundary value problems associated with Schr\"odinger operators, and presents a proof of Theorems \ref{main Robin} and \ref{main Robin stability}. In Section 5 we combine the work of Werner \cite{Werner1987} with the theory of boundary quadruples to state a Born rule for the time of detection.
\section{Dissipative Extensions of Skew-Symmetric Operators}
In this section, we review some applications of boundary quadruples presented in \cite{Arendt2023}. A Hilbert space $\mathcal{H}$ is taken to be a vector space over $\mathbb{C}$ equipped with an inner product $\langle \cdot, \cdot \rangle_{\mathcal{H}}$ that is linear in the first slot and anti-linear in the second slot, such that $\mathcal{H}$ is complete with respect to the norm $||\cdot||_{\mathcal{H}}:=\sqrt{\langle \cdot, \cdot \rangle_{\mathcal{H}}}$. An operator $A$ on $\mathcal{H}$ is a linear mapping defined on a subspace $D(A)$, called its domain, which takes values in $\mathcal{H}$. 
\begin{definition}
    \begin{enumerate}
        \item[(a)] An operator $B$ is an \textbf{extension} of $A$, denoted $A \subset B$, if $D(A)\subset D(B)$ and $B\psi=A\psi$ for $\psi \in D(A)$. 
        \item[(b)] An operator $A$ is called \textbf{dissipative} if 
        \begin{equation}
            \text{Re}\langle A\psi,\psi \rangle_{\mathcal{H}} \leq 0, \quad \forall \psi \in D(A). \end{equation}
        \item[(c)] An operator $A$ is called \textbf{skew-symmetric} if $\pm A$ are dissipative.
       
        \item[(d)] An operator $A$ is called \textbf{maximal dissipative} if $A\subset B$ with $B$ dissipative implies that $A=B$. 
        \item[(e)] An operator $A$ is called \textbf{m-dissipative} if $A$ is maximal dissipative and $D(A)$ is dense in $\mathcal{H}$. Equivalently, $A$ is \textbf{m-dissipative} if $A$ is dissipative and $(A-\mathds{1}):D(A)\to\mathcal{H}$ is invertible. 
        \item[(f)] A \textbf{$\mathbf{C_0}$ semigroup} $W_t$ on $\mathcal{H}$ is a one parameter family of linear operators on $\mathcal{H}$ for $t\in[0,\infty)$ satisfying
        \begin{enumerate}
            \item[1.] $W_0=\mathds{1}$, the identity operator on $\mathcal{H}$.
            \item[2.] $\forall t,s \geq 0$, $W_t W_s=W_{t+s}$, they form a semigroup under composition.
            \item[3.] $\forall \psi \in \mathcal{H}$,  $\lim_{t \to 0}||W_t\psi - \psi||_{\mathcal{H}} \to 0$.
        \end{enumerate}
       \item[(g)]  $W_t$ is a $\mathbf{C_0}$ \textbf{contraction semigroup} if additionally $||W_t \psi||_{\mathcal{H}}\leq ||\psi||_{\mathcal{H}}$ for all $\psi \in \mathcal{H}$, $t \geq 0$.
    \end{enumerate}
\end{definition}
The connection between $m$-dissipative operators and $C_0$ contraction semigroups is established by a theorem of Lumer and Phillips.
\begin{theorem}\label{Lumer-Phillips}\cite{Phi59}
    (Lumer-Phillips) Suppose $W_t$ is a $C_0$ contraction semigroup on some Hilbert space $\mathcal{H}$. Then there exists a linear operator $B$ such that
    \begin{enumerate}
        \item $D(B):= \{ \psi \in \mathcal{H}:\lim_{t \to 0^+}\frac{W_t \psi - \psi}{t} \text{ exists in }\mathcal{H} \}$ is dense in $\mathcal{H}$.
        \item $W_t=\exp(tB)$ i.e.  $\lim_{t \to t_0}\frac{W_t \psi - W_{t_0}\psi}{t-t_0}= B W_{t_0} \psi$ for all $\psi \in D(B)$, and all time $t_0\geq 0$.
        \item $W_t:D(B)\to D(B)$ and $W_tB\psi=BW_t\psi$ for all $\psi \in D(B)$, and all $t \geq 0$.
        \item $B$ is maximally dissipative, and hence also $m$-dissipative.
    \end{enumerate}
    The converse is also true, if $B$ is $m$-dissipative on $\mathcal{H}$ then it generates a $C_0$ contraction semigroup.
\end{theorem}
If $B$ is the generator of a $C_0$ contraction semigroup, then $2\text{Re}\langle B \psi, \psi \rangle_{\mathcal{H}}$ represents the rate at which $\exp(tB)\psi$ loses norm at $t=0$
\begin{equation}
    \frac{d}{dt} ||\exp(tB)\psi||_\mathcal{H}^2 \bigg{|}_{t=0}=\langle B \psi, \psi\rangle_{\mathcal{H}} + \langle \psi, B \psi \rangle_{\mathcal{H}}=2\text{Re}\langle B \psi, \psi \rangle_{\mathcal{H}}, \quad \forall \psi \in D(B).
\end{equation}
For a densely defined symmetric operator $\hat{H}_0$, our goal is to describe all $m$-dissipative extensions of $A_0=-i\hat{H}_0$, as these generate contraction semigroups that ``weakly" solve $i\frac{\partial \psi}{\partial t}=\hat{H}_0 \psi$. We begin by showing that all dissipative extensions of a densely defined skew-symmetric operator $A_0$ are restrictions of a ``maximal operator".
\begin{definition}
    For a densely defined operator $B$ on $\mathcal{H}$, we define its adjoint $B^*$
    \begin{equation}
        D(B^*):= \{\psi \in \mathcal{H}: \text{there exists some } \eta\in \mathcal{H} \text{ s.t } \langle \psi,B \phi \rangle_\mathcal{H}= \langle \eta,\phi \rangle_\mathcal{H}, \forall \phi \in D(B) \}, \quad B^*\psi:=\eta.
    \end{equation}
\end{definition}
\begin{proposition}\label{Diss ext rest}\cite[Theorem 2.5]{Arendt2023}
    Let $A_0$ be a densely defined skew-symmetric operator and $B$ a dissipative operator such that $A_0 \subset B$. Then $B \subset (-A_0)^*$.
\end{proposition}
Any dissipative extension of $A_0$ must be given by some restriction of $(-A_0)^*$, so we would like to identify which class of restrictions return $m$-dissipative operators. To do this we apply the theory of boundary quadruples. For convenience we henceforth denote $A:=(-A_0)^*$.
\begin{definition}
    A \textbf{boundary quadruple} for a densely defined skew-symmetric operator $A_0$ consists of Hilbert spaces $\mathcal{H}_\pm$ and continuous, surjective linear maps $G_\pm: D(A)\to \mathcal{H}_\pm$ satisfying 
    \begin{equation}
        \langle A \psi, \phi \rangle_{\mathcal{H}} + \langle \psi, A \phi \rangle_{\mathcal{H}} = \langle G_+\psi, G_+ \phi \rangle_{\mathcal{H}_+} - \langle G_- \psi, G_- \phi \rangle_{\mathcal{\mathcal{H}_-}} 
    \end{equation}
    for all $\psi, \phi \in D(A)$, and 
    \begin{equation}
        \ker G_+ + \ker G_- = D(A), \text{ or equivalently } (G_+,G_-):D(A)\to \mathcal{H}_+ \times \mathcal{H}_- \text{ surjective}.
    \end{equation}
\end{definition}
Boundary quadruples always exist (see \cite[Example 3.6 and Example 3.7]{Arendt2023}), and all boundary quadruples are isometric (see \cite[Theorem 3.16]{Arendt2023}), but they may not obviously correspond to the behavior of wave functions along the boundary of some region. The example below does.
\begin{example}\label{One D Boundary Quad Example}
\em    Consider $A_0=\partial_x$ densely defined on $C_c^\infty((0,1))\subset L^2((0,1))$. This operator is skew symmetric, and its maximal extension is  $A=\partial_x :H^1((0,1))\to L^2((0,1))$. Recall by Sobolev embedding theorems \cite{Brezis} that $H^1((0,1))\subset C([0,1])$, the space of continuous functions on $[0,1]$, so the values of $H^1$ functions are well-defined at the endpoints. Integration by parts returns
    \begin{equation}
        \langle \partial_x \psi, \phi \rangle_{L^2((0,1))} + \langle \psi, \partial_x \phi \rangle_{L^2((0,1))} = \psi(1)\phi^*(1) - \psi(0)\phi^*(0), \quad \forall \psi, \phi \in H^1([0,1]).
    \end{equation}
    A boundary quadruple for $A$ is thus given by $G_\pm: H^1((0,1))\to \mathbb{C}$ with $G_+ \psi:= \psi(1)$, $G_-\psi:=\psi(0)$, after checking that $(G_+,G_-):H^1((0,1))\to \mathbb{C}\times \mathbb{C}$ surjective.
\end{example}
Given a fixed boundary quadruple $(\mathcal{H}_\pm, {G}_\pm)$ for a densely defined skew symmetric operator $A_0$, we will now parameterize all of its $m$-dissipative extensions. It is useful to first state a lemma relating the boundary quadruple to $\overline{A_0}$. Recall that any densely-defined skew-symmetric operator $A_0$ is closeable, and its closure $\overline{A_0}$ is also skew-symmetric.
\begin{lemma}\label{Closure} \cite[Proposition 3.9]{Arendt2023}
$D(\overline{A_0})=\ker G_+ \cap \ker G_-$ and $\overline{A_0}\psi=A\psi$ for all $\psi \in D(\overline{A_0})$.
\end{lemma}
Like $\overline{A_0}$, the domain of any $m$-dissipative extension of $A_0$ can be easily expressed in terms of the boundary quadruple. For $\Phi:\mathcal{H}_- \to \mathcal{H}_+$ a linear contraction, we define the operator $A_\Phi$ on $\mathcal{H}$ via
\begin{equation}
    D(A_\Phi):= \{\psi \in D(A): \Phi G_- \psi= G_+ \psi \}, \quad A_\Phi \psi:= A\psi.
\end{equation}
Clearly $A_0 \subset A_\Phi \subset A$, and $A_\Phi$ dissipative because for all $\psi \in D(A_\Phi)$
\begin{equation}\label{Probability Contracting}
    2\text{Re}\langle A_\Phi \psi, \psi \rangle_{\mathcal{H}}= ||G_+\psi||_{\mathcal{H}_+}^2 - ||G_-\psi||_{\mathcal{H}_-}^2=||\Phi G_-\psi||_{\mathcal{H}_+}^2 - ||G_-\psi||_{\mathcal{H}_-}^2 \leq 0.
\end{equation}
We shall now state a parameterization of all $C_0$ contraction semigroup generators that extend $A_0$.
\begin{theorem}\label{Description}
    \cite[Theorem 3.10]{Arendt2023}  Let $B$ be an operator on $\mathcal{H}$ such that $A_0 \subset B$. Then the following are equivalent
    \begin{enumerate}
        \item[(a)] $B$ is $m-$dissipative.
        \item[(b)] There exists a linear contraction $\Phi:\mathcal{H_-}\to \mathcal{H}_+$ such that $B=A_\Phi$.
    \end{enumerate}
\end{theorem}
When the boundary quadruple mappings $G_\pm$ directly relate to the values of the wave function along the boundary $\partial \Omega$ of some region $\Omega$, as was the case in Example \ref{One D Boundary Quad Example}, the condition $G_+\psi=\Phi G_-\psi$ becomes identifiable as an ``absorbing boundary condition". In combination with the theorem of Lumer-Phillips, Theorem \ref{Description} states that for every linear contraction $\Phi:\mathcal{H}_- \to \mathcal{H}_+$, the initial-boundary value problem
\begin{equation}\label{abstract IBVP}
    \left\{\begin{array}{rclr}
       \partial_t \psi&=& A \psi \quad 
       \\
       \psi\big{|}_{t=0} &=&\psi_0 
       \\
    G_+\psi&=&\Phi G_- \psi
    \end{array}
    \right.
\end{equation}
admits a unique global-in-time solution for each initial $\psi_0 \in D(A_\Phi)$, and the solution mappings $W_t:\psi_0\mapsto\psi_t$ extend to a $C_0$ contraction semigroup on $\mathcal{H}$. Theorem \ref{Description} states that a converse is also true, any $C_0$ contraction semigroup whose generator is an extension of $A_0$ (or equivalently, a restriction of $-A_0^*$) must be the solution mapping of an initial-boundary value problem of the form (\ref{abstract IBVP}).
\par
When the boundary condition is \textit{strictly} absorbing; $||\Phi \xi||_{\mathcal{H}_+}<||\xi||_{\mathcal{H}_-}$ for all $\xi \neq 0$ in $\mathcal{H}_-$; one may expect $\exp(tA_\Phi)\psi$ to asymptotically vanish as $t \to \infty$. This is generally not true, but the Theorem below provides necessary and sufficient conditions when $A_\Phi$ has compact resolvent. Note that since $A_\Phi$ is $m$-dissipative its resolvent set $\rho(A_\Phi)$ is non-empty with $1 \in \rho(A_\Phi)$.  
\begin{theorem}\label{asympt stability}
    Let $\Phi:\mathcal{H}_-\to \mathcal{H}_+$ be a linear contraction satisfying $||\Phi \xi||_{\mathcal{H}_+}<||\xi||_{\mathcal{H}_-}$ for all $0\neq \xi \in \mathcal{H}_-$, and suppose $(A_\Phi - \alpha)^{-1}$ is a compact operator on $\mathcal{H}$ for some $\alpha \in \rho(A_\Phi)$. Then the following are equivalent
    \begin{enumerate}
        \item $\overline{A_0}$ has no eigenvalues on the imaginary axis.
        \item $||\exp(tA_\Phi)\psi_0||_\mathcal{H}\xrightarrow{t \to \infty}0$ for each $\psi_0 \in \mathcal{H}$.
    \end{enumerate}
\end{theorem}
\begin{proof}
The first direction is straightforward. Suppose $\overline{A_0}$ admits an eigenvalue $i\lambda$ on the imaginary axis, and let $0\neq\psi_\lambda\in D(\overline{A_0})$ satisfy $\overline{A_0}\psi_{\lambda}=i\lambda \psi_{\lambda}$. Then $A_\Phi \psi_\lambda=i\lambda \psi_{\lambda}$, hence $\exp(tA_\Phi)\psi_\lambda=e^{i\lambda t}\psi_\lambda$, which does not asymptotically vanish as $t \to \infty$.

To prove the other direction, we rely on a theorem of Arendt, Batty, Lyubich and Phong.
\begin{theorem*}
    \cite{ArendtStab,Lyubich1988} Let $W_t=\exp(tB)$ be a $C_0$ contraction semigroup on $\mathcal{H}$ such that $B$ has no eigenvalues on the imaginary axis and $\sigma(B)\cap i\mathbb{R}$ is countable. Then $||W_t \psi_0||\xrightarrow{t \to \infty}0$ for each $\psi_0 \in \mathcal{H}$. 
\end{theorem*}
We aim to show that $A_\Phi$ has no eigenvalues on the imaginary axis and $\sigma(A_\Phi)\cap i \mathbb{R}$ is countable. The second condition follows almost immediately from our assumption that $A_\Phi$ has compact resolvent.
\begin{lemma}
    If $B$ is a densely defined operator with $(B-\lambda)^{-1}$ compact for some $\lambda \in \rho(B)$. Then $\sigma(B)$ consists only of countably many eigenvalues.
\end{lemma}
\begin{proof}
When $\mathcal{H}$ is finite dimensional this proof becomes trivial, so take $\dim(\mathcal{H})=\infty$. Since $(B-\lambda)^{-1}$ is a compact operator on $\mathcal{H}$, the spectral theorem for compact operators \cite[Theorem 6.8]{Brezis} states that $\sigma((B-\lambda)^{-1})$ consists only of countably many eigenvalues and $0$. Take $0 \neq\mu\in \rho((B-\lambda)^{-1})$. Then
\begin{equation}
    \frac{1}{\mu}-(B-\lambda)=\frac{1}{\mu} (B-\lambda)\left((B-\lambda)^{-1}-\mu \right).
\end{equation}
So $\frac{1}{\mu}\in \rho(B-\lambda)$ with 
\begin{equation}
    \left(\frac{1}{\mu}-(B-\lambda)\right)^{-1}=\mu\left((B-\lambda)^{-1}-\mu \right)^{-1}(B-\lambda)^{-1}.
\end{equation}
Hence for any $\nu \in \sigma(B-\lambda)$ we have $\nu \neq 0$ (since $B-\lambda$ is invertible) and $\frac{1}{\nu}\in \sigma((B-\lambda)^{-1})$, in particular $\frac{1}{\nu}$ must be an eigenvalue of $(B-\lambda)^{-1}$. But this implies $\nu$ is an eigenvalue of $B-\lambda$, so the spectrum of $B-\lambda$ only consists of countably many eigenvalues, and consequently $\sigma(B)$ only consists of countably many eigenvalues.
\end{proof}
Returning to our proof of Theorem \ref{asympt stability}, it is now sufficient to prove that $A_\Phi$ has no eigenvalues on the imaginary axis whenever $\overline{A_0}$ has no eigenvalues on the imaginary axis. We proceed by contrapositive, suppose $A_\Phi$ has an eigenvalue $i\lambda$ along the imaginary axis, and let $\psi_\lambda\in D(A_\Phi)$ be a corresponding eigenvector. Then 
\begin{equation}
   ||\Phi G_-\psi_\lambda||_{\mathcal{H}_+}^2- ||G_-\psi_\lambda||_{\mathcal{H}_-}^2=2\text{Re}\langle A\psi_\lambda, \psi_\lambda \rangle_{\mathcal{H}}=2\text{Re}\langle i\lambda \psi_\lambda,\psi_\lambda \rangle_{\mathcal{H}}=0. 
\end{equation}
Since $||\Phi G_- \psi_\lambda||_{\mathcal{H}_+}=||G_-\psi_\lambda||_{\mathcal{H}_-}$ and $\Phi$ is \textit{strictly} contractive, we must have $G_+\psi_\lambda=\Phi G_- \psi_\lambda=G_-\psi_\lambda=0$. By Lemma \ref{Closure}, $\psi_\lambda \in D(\overline{A_0})$ and $\overline{A_0}\psi_\lambda=A_\Phi \psi_\lambda=i\lambda\psi_\lambda$, so $\overline{A_0}$ must also have eigenvalues along the imaginary axis. 
\end{proof}
\begin{remark}
    Theorem \ref{asympt stability} is typically useful when $A_0$ is a differential operator on a bounded Lipschitz domain $\Omega$ with domain $D(A_0)=C_c^\infty(\Omega)$ dense in $L^2(\Omega)$. In practice $D(A_\Phi)$ is often (but not always) a subset of $H^s(\Omega)$ for some $s>0$, in which case $A_\Phi$ has compact resolvent on $L^2(\Omega)$.
\end{remark}
\section{Schr\"odinger Operators on Bounded $C^2$ Domains}
In this section, we parameterize all $m$-dissipative extensions of the densely defined skew-symmetric operator $-i\hat{H}_0=-i(-\Delta + V): C_c^\infty(\Omega)\to L^2(\Omega)$, where $\Omega$ is some bounded $C^2$ region in $\mathbb{R}^n$ and $V$ a real valued and bounded potential. Parameterization results for second-order elliptic differential operators can be found in the literature (see, e.g. \cite{Gesztesy11,BEHRNDT2014}) and go back to the classical work of Birman \cite{Bir56}, Visik \cite{Visik}, and Grubb \cite{Grubb1968ACO}. However, we will find that taking the approach of Wegner \cite{Wegner2017} and Arendt et al. \cite{Arendt2023} simplifies the parameterization to a form that is well adapted to defining detection time distributions.  

Before tackling the higher dimensional case, let us first review the parameterization result in a simple one dimensional setting.
\subsection{One Dimensional Case}
For simplicity we only consider the case of a bounded interval $\Omega=(a,b)\subset \mathbb{R}$. We define the \textit{pre-minimal} operator $\hat{H}_0=-\partial_x^2 + V:C_c^\infty((a,b))\to L^2((a,b))$, where $V\in L^\infty((a,b),\mathbb{R})$. The closure of the pre-minimal operator is called the \textit{minimal} operator $\hat{H}$ and has domain $D(\hat{H})=H^2_0((a,b))$, while the maximal operator $\hat{H}^*=(-\partial_x^2 + V)$ has domain $D(\hat{H}^*)=H^2((a,b))$. For $\psi \in H^2((a,b))$ the Sobolev embedding theorems \cite{Brezis} guarantee $\psi$ and $\partial_x \psi$ admit continuous representatives in $C([a,b])$. We define the mappings $G_\pm:H^2((a,b))\to \mathbb{C}^2$
\begin{equation}\label{One D BQ}
    G_+\psi:=\frac{1}{\sqrt{2}}\begin{pmatrix}
        \psi(a)-i \partial_x \psi(a)
        \\
        \psi(b)+i \partial_x \psi(b)
    \end{pmatrix}, \quad G_-\psi:=\frac{1}{\sqrt{2}}\begin{pmatrix}
        \psi(a)+i \partial_x \psi(a)
        \\
        \psi(b)-i \partial_x \psi(b)
    \end{pmatrix}.
\end{equation}
It is then straightforward to apply integration by parts and show
\begin{equation}\label{One D Flux}
    \langle -i\hat{H}^*\psi, \phi \rangle_{L^2((a,b))} + \langle \psi, -i\hat{H}^*\phi \rangle_{L^2((a,b)}= \langle G_+ \psi, G_+ \phi \rangle_{\mathbb{C}^2} - \langle G_- \psi, G_- \phi \rangle_{\mathbb{C}^2}, \quad \forall \psi, \phi \in H^2((a,b)).
\end{equation}
We now state the parameterization result for Schr\"odinger operators in one space dimension.
\begin{theorem}\cite[Theorem 6.5]{Arendt2023} Let $(a,b)\subset \mathbb{R}$ be a bounded interval, and let $\hat{H}:=(-\partial_x^2 +V)\big{|}_{H^2_0((a,b))}$ with $V \in L^\infty((a,b),\mathbb{R})$. Then the maps $G_\pm:H^2((a,b))\to \mathbb{C}^2$ define a boundary quadruple for $-i\hat{H}$. Consequently, a $C_0$ contraction semigroup $W_t$ on $L^2((a,b))$ has its generator extended by $-i\hat{H}^*$ if and only if there exists a linear contraction $\Phi:\mathbb{C}^2\to \mathbb{C}^2$ such that $W_t=\exp(-it\hat{H}_\Phi)$ with
    \begin{equation}
        D(\hat{H}_\Phi):= \{ \psi \in H^2((a,b)): G_+\psi=\Phi G_-\psi\}, \quad \hat{H}_\Phi\psi:=(-\partial_x^2 + V)\psi.
    \end{equation}
\end{theorem}
\begin{proof}
By Theorem \ref{Lumer-Phillips}, Theorem \ref{Description}, and equation (\ref{One D Flux}), it suffices to show
\begin{equation}
    \ker G_- + \ker G_+= H^2((a,b)).
\end{equation}
Let $\psi \in H^2((a,b))$, and choose $\xi, \phi \in C^2((a,b))$ such that
\begin{equation}
    \phi(a)=-i\partial_x\psi(a), \quad \partial_x \phi(a)=i\psi(a), \quad \phi(b)=\partial_x \phi(b)=0.
\end{equation}
\begin{equation}
    \xi(b)=-i\partial_x\psi(b), \quad \partial_x \xi(b)=i\psi(b), \quad \xi(a)=\partial_x \xi(a)=0.
\end{equation}
Then
\begin{equation}
    \frac{1}{2}(\psi + \phi -\xi)\in \ker G_+, \text{ and } \frac{1}{2}(\psi-\phi +\xi)\in \ker{G_-},
\end{equation}
so the sum $\psi \in \ker G_+ + \ker G_-$.
\end{proof}
\begin{remark}
    This simple boundary quadruple construction fails for Schr\"odinger operators in dimensions $n\geq 2$.  This is because in higher space dimensions the domain of the maximal operator $D(\hat{H}^*)$ is not contained within any Sobolev space $H^s(\Omega)$ for $s>0$. So, one cannot rely on the standard Sobelev embedding theorems to make sense of the values of $\psi$ and its normal derivative along the boundary $\partial \Omega$.
\end{remark}
The following proposition reviews some classical examples of boundary conditions and their associated extensions of $\hat{H}$.
\begin{proposition}
     Setting $\Phi_N= \mathds{1}$, the identity on $\mathbb{C}^2$, returns the \textbf{Neumann} extension $\hat{H}_N$ of $\hat{H}$, with domain
    \begin{equation}
        D(\hat{H}_N)=\{ \psi \in H^2((a,b)): \partial_x \psi (a)=\partial_x \psi(b)=0 \}.
    \end{equation}
     Setting $\Phi_D= -\mathds{1}$ returns the \textbf{Dirichlet} extension $\hat{H}_D$ of $\hat{H}$, with domain
    \begin{equation}
        D(\hat{H}_D)=\{ \psi \in H^2((a,b)): \psi (a)= \psi(b)=0 \}.
    \end{equation}
     Setting $\Phi_P= \begin{pmatrix}
            0 & 1
            \\
            1 & 0
        \end{pmatrix}$ returns the \textbf{periodic} extension $\hat{H}_P$ of $\hat{H}$, with domain 
    \begin{equation} 
        D(\hat{H}_P)=\{ \psi \in H^2((a,b)): \psi (a)= \psi(b), ~\partial_x \psi(a)=\partial_x \psi(b) \}.
    \end{equation}
\end{proposition}
    \begin{remark}
     The evolutions generated by these extensions preserve the total probability in $\Omega$ and are thus not valid as theoretical models for quantum particles undergoing detection. The last example additionally gives rise to \textit{non-local} dynamics, with probability allowed to exit through one end-point and be instantly transported to the other.
    \end{remark}
    \begin{proposition}
     $\hat{H}_\Phi$ generates local dynamics in $(a,b)$ if and only if $\Phi$ is of the form $\Phi=\begin{pmatrix}
         \Phi_a & 0
         \\
         0 & \Phi_b
     \end{pmatrix}$ with $|\Phi_a|, |\Phi_b| \leq 1$. If $\Phi_a \neq -1 \neq \Phi_b$ then $H_\Phi$ is a \textbf{local Robin} extension \begin{equation}
        D(\hat{H}_\Phi)=\{ \psi \in H^2((a,b)): -\partial_x \psi (a)=i \frac{1-\Phi_a}{1+\Phi_a} \psi(a),~ \partial_x \psi(b)=i\frac{1-\Phi_b}{1+\Phi_b} \psi(b) \}.
    \end{equation}
    If $|\Phi_a|$ and $|\Phi_b|$ are strictly less than $1$, then by Theorem \ref{asympt stability} the evolution generated by $\hat{H}_\Phi$ is asymptotically stable, $||\exp(-it\hat{H}_\Phi)\psi_0||_{L^2((a,b))}\xrightarrow{t \to \infty}0$ for each $\psi_0\in L^2((a,b))$. 
    \end{proposition}
    
\subsection{Higher Dimensions}
Fix $n \geq 2$ and let $\Omega$ be a bounded domain in $\mathbb{R}^n$ with $C^2$ boundary $\partial \Omega$. In this section we seek to describe all $m-$dissipative extensions of the Schr\"odinger operator $-i\hat{H}_0=-i(-\Delta + V):C_c^\infty(\Omega)\to L^2(\Omega)$, where $V\in L^\infty(\Omega,\mathbb{R})$. We accomplish this by constructing a boundary quadruple for the maximal operator $-i \hat{H}_0^*$. This section will closely follow the boundary triple construction presented in chapter $9$ of \cite{BHS2020}.
\par
We denote the Sobolev space of order $k\in \mathbb{R}$ on $\Omega$ by $H^k(\Omega)$, and the closure of $C_c^\infty(\Omega)$ in $H^k(\Omega)$ is denoted $H_0^k(\Omega)$. Sobolev spaces on the boundary are denoted $H^s(\partial \Omega)$ for $s \in  \mathbb{R}$. For $s>0$ the dual space pairing $\langle \cdot, \cdot \rangle_{H^{-s}(\partial \Omega)\times H^s(\partial \Omega)}$ is anti-linear in the second slot and satisfies
\begin{equation}
    \langle \xi, \chi \rangle_{H^{-s}(\partial \Omega)\times H^s(\partial \Omega)}=\begin{cases} \langle \xi, \chi\rangle_{L^2(\partial \Omega)}, \quad & \xi \in L^2(\partial \Omega)
    \\
    \langle \xi, \chi\rangle_{H^{-t}(\partial \Omega)\times H^{t}(\partial \Omega)}, \quad & \xi \in H^{-t}(\partial \Omega), ~0 \leq t\leq s
    \end{cases}
\end{equation} for $\chi \in H^s(\partial \Omega)$. Let $\iota_\pm:H^{\pm 1/2}(\partial \Omega) \to L^2(\partial \Omega)$ denote the isometric isomorphisms such that
\begin{equation}
    \langle \xi, \chi \rangle_{H^{-1/2}(\partial \Omega)\times H^{1/2}(\partial \Omega)}= \langle \iota_- \xi, \iota_+ \chi\rangle_{L^2(\partial \Omega)}.
\end{equation}
\begin{remark}
    On $\mathbb{R}^n$ the isometries $\iota_\pm:H^{\pm 1/2}(\mathbb{R}^n)\to L^2(\mathbb{R}^n)$ can be represented in terms of the Fourier transform $\mathcal{F}$ as $\iota_\pm=\mathcal{F}^{-1}(1+|\cdot|^2)^{\pm 1/4}\mathcal{F}$.
\end{remark}
For $s\in [0,3/2]$ the restrictions 
\begin{equation}
    \iota_+: H^{s+1/2}(\partial \Omega)\to H^s(\partial \Omega)
\end{equation}
and 
\begin{equation}
    \iota_-:H^{s}(\partial \Omega)\to H^{s+1/2}(\partial \Omega)
\end{equation}
are isometric isomorphisms such that $\iota_+\iota_- \xi=0$ for all $\xi \in H^s(\partial \Omega)$ and $\iota_- \iota_+ \chi=\chi$ for all $\chi \in H^{s+1/2}(\partial \Omega)$.
\par
Denoting the unit normal vector field pointing outwards of $\Omega$ by $\partial_n$, we recall a classical result for the trace operator of $H^2(\Omega)$ functions to the boundary.
\begin{lemma}\label{Trace Definition}
    \cite[Theorem 8.3]{LH1972}\cite[Theorem 1.5.1.2]{Grisvard} For $\Omega\subset \mathbb{R}^n$ a bounded $C^2$ domain, the trace map $\psi \mapsto \left(\psi\big{|}_{\partial \Omega}, \partial_n \psi \big{|}_{\partial \Omega} \right)$ defined for $\psi \in C^\infty(\Omega) \to H^{{3/2}}(\partial \Omega)\times H^{{1/2}}(\partial \Omega)$ admits a continuous extension $\psi \mapsto \left(\tau_D \psi, \tau_N \psi \right)$ that is surjective for $H^2(\Omega)\to H^{3/2}(\partial \Omega)\times H^{{1/2}}(\partial \Omega)$ and admits a continuous right inverse. 
\end{lemma}
We call $\tau_D$ and $\tau_N$ the Dirichlet and Neumann operator respectively. These operators can be used to prove a Green's identity for the Laplacian.
\begin{lemma*}\label{Green's Identity}\textit{(Green's identity).}
    For $u,v \in H^2(\Omega)$, one can integrate by parts to show
    \begin{equation}
        \langle i\Delta u,v\rangle_{L^2(\Omega)}+ \langle u,i\Delta v \rangle_{L^2(\Omega)}=i\left( \langle \tau_N u, \tau_D v\rangle_{L^2(\partial \Omega)} -\langle  \tau_D u, \tau_N v\rangle_{L^2(\partial \Omega)}\right).
    \end{equation}
\end{lemma*}
To construct a boundary quadruple for $-i\hat{H}_0^*$, we wish to extend these trace operators to $D(-i\hat{H}_0^*)$ and prove a similar Green's identity.
\par
First, let us recall that the \textit{preminimal} operator for $-\Delta+V$ is defined on $D(\hat{H}_0)={C_c^\infty(\Omega)}$. The closure of $\hat{H}_0$ is the \textit{minimal} operator $\hat{H}$, and has domain $D(\hat{H})=H^2_0(\Omega)$. The \textit{maximal} operator is the adjoint of the \textit{minimal operator}, i.e $\hat{H}^*$. We must point out that $H^2(\Omega)\subset D(\hat{H}^*)$ but the two sets are not equal, so some care must be taken when defining $(\psi, \partial_n \psi)\big{|}_{\partial \Omega}$ for $\psi \in D(\hat{H}^*)$.
\par
In addition to the minimal and maximal operators, our boundary quadruple construction will frequently refer to two closed extensions of $\hat{H}$.
\begin{lemma}
    Let $\Omega\subset \mathbb{R}^n$ be a $C^2$ bounded region, and let $\hat{H}:=(-\Delta+V)\big{|}_{H^2_0(\Omega)}$ with $V \in L^\infty(\Omega,\mathbb{R})$. Then the \textit{Dirichlet} extension $\hat{H}_D$ defined by
    \begin{equation}
    D(\hat{H}_D):= \{ \psi \in H^2(\Omega):  \tau_D \psi=0\}, \quad \hat{H}_D \psi:=(-\Delta+V)\psi
\end{equation}
is a closed symmetric operator on $L^2(\Omega)$. Similarly, the Neumann extension $\hat{H}_N$ defined by 
    \begin{equation}
    D(\hat{H}_N):= \{ \psi \in H^2(\Omega):  \tau_N \psi=0\}, \quad \hat{H}_N \psi:=(-\Delta+V)\psi.
\end{equation}
is a closed symmetric operator on $L^2(\Omega)$.
\end{lemma}
These operators are useful in allowing us to write the domain of the maximal operator as a direct sum of familiar spaces.
\begin{lemma}\label{Direct Sum}
    For any $\lambda \in \rho(\hat{H}_D)$ we have the direct sum decompositions
    \begin{equation}\label{Dirichlet Star Decomp}
        D(\hat{H}^*)= \ker \tau_D \oplus \ker(\hat{H}^*-\lambda)
    \end{equation}
\begin{equation}\label{Dirichlet Decomp}
        H^2(\Omega)= \ker \tau_D \oplus \left\{ \psi_{\lambda} \in H^2(\Omega):(\hat{H}^*-\lambda)\psi_{\lambda}=0\right\}.
    \end{equation}
\end{lemma}
\begin{proof}
    Let $\psi \in D(\hat{H}^*)$. By the invertibility of $(\lambda-\hat{H}_D):H^2(\Omega) \to L^2(\Omega)$, there exists a unique $\psi_D\in H^2(\Omega)$ such that $(\hat{H}_D-\lambda)\psi_D=(\hat{H}^*-\lambda)\psi$. Since $\hat{H}_D\subset \hat{H}^*$, it follows that $\psi_\lambda:=\psi-\psi_D$ satisfies $(\hat{H}^*-\lambda)\psi_\lambda=0$. Hence $\psi=\psi_D+\psi_\lambda \in D(\hat{H}_D)\oplus \ker(\hat{H}^*-\lambda)$, and the proof follows from $D(\hat{H}_D)=\ker\tau_D$.
\end{proof}
These decompositions are used to extend the trace operators $\tau_D$, $\tau_N$ to the domain of the maximal operator, and play a crucial role in writing down an ``abstract Green's identity" for $-i\hat{H}^*$.
\begin{lemma}\label{Trace Extensions} \cite[Theorem 8.3.9]{BHS2020}
    The Dirichlet and Neumann trace operators $\tau_D:H^2(\Omega)\to H^{3/2}(\partial \Omega)$, $\tau_N:H^2(\Omega)\to H^{1/2}(\partial \Omega)$ admit continuous and surjective extensions
    \begin{equation}
        \tilde{\tau}_D:D(\hat{H}^*)\to H^{{-1/2}}(\partial \Omega), \quad \tilde{\tau}_N:D(\hat{H}^*)\to H^{{-3/2}}(\partial \Omega)
    \end{equation}
    In addition,
    \begin{equation}
        \ker (\tilde{\tau}_D)=\ker(\tau_D)=D(\hat{H}_D), \quad \ker (\tilde{\tau}_N)=\ker(\tau_N)=D(\hat{H}_N).
    \end{equation}
\end{lemma}
The extended trace operators give rise to an extended Green's identity for elements in $\hat{H}^*$.
\begin{corollary}\cite[Corollary 8.3.11]{BHS2020}
    The Green's identity can be extended to 
    \begin{equation}\label{Green's Identity Star}
        \langle -i\hat{H}^*\psi,v\rangle_{L^2(\Omega)}+ \langle \psi,-i\hat{H}^*v \rangle_{L^2(\Omega)}=i\left( \langle \tilde{\tau}_N \psi, \tau_D v\rangle_{H^{-3/2}(\partial \Omega)\times H^{3/2}(\partial \Omega)} -\langle  \tilde{\tau}_D \psi, \tau_N v\rangle_{H^{-1/2}(\partial \Omega)\times H^{1/2}(\partial \Omega)}\right).
    \end{equation}
    for $\psi \in D(\hat{H}^*)$, $v \in H^2(\Omega)$.
\end{corollary}
The formula above does not seem sufficient for our purposes, as we require a Green's identity in the case that both $\psi$ and $v$ are in $D(\hat{H}^*)$. Luckily, it is possible to extend this identity by decomposing the wave functions appropriately. 
\begin{proposition}\label{Current Decomposition}
        Let $\eta \in \rho(\hat{H}_D)\cap \mathbb{R}$, which exists since $\hat{H}_D-V$ is a positive operator and $V \in L^\infty(\Omega,\mathbb{R})$. Also let $\iota_\pm:H^{\pm {1/2}}(\partial \Omega)\to L^2(\partial \Omega)$ be the isomorphisms such that $\langle \xi,\chi \rangle_{H^{-1/2}(\partial \Omega) \times H^{1/2}(\partial \Omega)}=\langle \iota _- \xi, \iota_+ \chi \rangle_{L^2(\partial \Omega)}$. Then, defining $G_\pm(\eta): D(-i\hat{H}^*)\to L^2(\partial \Omega)$ by
    \begin{equation}
        G_+\psi:=\frac{1}{\sqrt{2}}\left(\iota_- \tilde{\tau}_D \psi  + i \iota_+ \tau_N \psi_D\right)
    \end{equation}
    \begin{equation}
        G_-\psi:=\frac{1}{\sqrt{2}}\left(\iota_-\tilde{\tau}_D \psi  - i \iota_+ \tau_N  \psi_D \right)
    \end{equation}
    where $\psi=\psi_D + \psi_\eta$ is decomposed according to equation (\ref{Dirichlet Star Decomp}) returns a further extension of the Green's identity for elements in $\psi, \phi \in D(\hat{H}^*)$
    \begin{equation}\label{Current Star Decomp Eq}
        \langle -i\hat{H}^* \psi,\phi \rangle_{L^2(\Omega)} + \langle \psi,-i\hat{H}^* \phi \rangle_{L^2(\Omega)}=\langle G_+\psi,G_+\phi \rangle_{L^2(\partial \Omega)} - \langle G_-\psi,G_- \phi \rangle_{L^2(\partial \Omega)}.
    \end{equation}
        \end{proposition}
      \begin{proof}
          Let $\psi,\phi \in D(\hat{H})$, and decompose $\psi= \psi_D + \psi_\eta$, $\phi=\phi_D + \phi_\eta$ according to equation (\ref{Dirichlet Star Decomp}). First, since $\hat{H}_D$ is self-adjoint 
          \begin{equation}
           \langle -i\hat{H}^* \psi_D,\phi_D \rangle_{L^2(\Omega)} + \langle \psi_D,-i\hat{H}^* \phi_D \rangle_{L^2(\Omega)}=\langle -i\hat{H}_D \psi_D,\phi_D \rangle_{L^2(\Omega)} + \langle \psi_D,-i\hat{H}_D \phi_D \rangle_{L^2(\Omega)}=0.
          \end{equation}
          Similarly, since $\eta$ is real we have 
          \begin{equation}
           \langle -i\hat{H}^* \psi_\eta,\phi_\eta \rangle_{L^2(\Omega)} + \langle \psi_\eta,-i\hat{H}^* \phi_\eta \rangle_{L^2(\Omega)}=\langle -i\eta\psi_\eta,\phi_\eta \rangle_{L^2(\Omega)} + \langle \psi_\eta,-i\eta \phi_\eta \rangle_{L^2(\Omega)}=0.
           \end{equation}
           Hence
\begin{equation}
    \begin{split}
             \langle -i\hat{H}^* \psi,\phi \rangle_{L^2(\Omega)} + \langle \psi,-i\hat{H}^* \phi \rangle_{L^2(\Omega)}&=\langle -i\hat{H}^* \psi_\eta,\phi_D \rangle_{L^2(\Omega)} +\langle \psi_\eta,-i\hat{H}^* \phi_D \rangle_{L^2(\Omega)}
             \\
             &+ \langle -i\hat{H}^* \psi_D,\phi_\eta \rangle_{L^2(\Omega)} +\langle \psi_D,-i\hat{H}^* \phi_\eta \rangle_{L^2(\Omega)}.
    \end{split}
\end{equation}
Since $\psi_D, \phi_D \in H^2(\Omega)$, we may apply the generalized Green's identity $(\ref{Green's Identity Star})$ to the two pairs of terms. Applying the identity along with $\tau_D \psi_D= \tau_D \phi_D=0$ returns
\begin{equation}\label{Gen Green's Identity Inner Prod}
\begin{split}
    \langle -i\hat{H}^* \psi,\phi \rangle_{L^2(\Omega)} + \langle \psi,-&i\hat{H}^* \phi \rangle_{L^2(\Omega)}=
    \\
    &=i\left( \langle \tau_N \psi_D, \tilde{\tau}_D \phi_\eta \rangle_{H^{1/2}(\partial \Omega)\times H^{-1/2}(\partial \Omega)} -\langle  \tilde{\tau}_D \psi_\eta, \tau_N \phi_D \rangle_{H^{-1/2}(\partial \Omega)\times H^{1/2}(\partial \Omega)}\right)
    \\
    &=i\left( \langle \tau_N \psi_D, \tilde{\tau}_D \phi \rangle_{H^{1/2}(\partial \Omega)\times H^{-1/2}(\partial \Omega)} -\langle  \tilde{\tau}_D \psi, \tau_N \phi_D \rangle_{H^{-1/2}(\partial \Omega)\times H^{1/2}(\partial \Omega)}\right)
    \\
    &=i\left( \langle \iota_+ \tau_N \psi_D, \iota_- \tilde{\tau}_D \phi \rangle_{L^2(\partial \Omega)} -\langle  \iota_-\tilde{\tau}_D \psi, \iota_+ \tau_N \phi_D \rangle_{L^2(\partial \Omega)}\right).
    \end{split}
\end{equation}
It is then not difficult to compute
    \begin{equation}
        \langle G_+\psi, G_+ \phi\rangle_{L^2(\partial \Omega)}- \langle G_-\psi, G_- \phi\rangle_{L^2(\partial \Omega)}= i\left( \langle \iota_+ \tau_N \psi_D, \iota_- \tilde{\tau}_D \phi \rangle_{L^2(\partial \Omega)} -\langle  \iota_-\tilde{\tau}_D \psi, \iota_+ \tau_N \phi_D \rangle_{L^2(\partial \Omega)}\right).
    \end{equation}
This concludes our proof of Proposition \ref{Current Decomposition}. \end{proof}
We now formally restate and prove the first main result of this paper.
\begin{theoremmainagain}\label{Parameterization Schrod}
    Let $\Omega \subset \mathbb{R}^n$ be a bounded $C^2$ domain, and let $\hat{H}:=(-\Delta+V)\big{|}_{H^2_0(\Omega)}$ with $V \in L^\infty(\Omega,\mathbb{R})$. Then for any $\eta\in \rho(\hat{H}_D)\cap \mathbb{R}$, the maps $G_\pm(\eta):D(\hat{H}^*)\to L^2(\partial \Omega)$ defined in Proposition \ref{Current Decomposition} define a boundary quadruple for $-i\hat{H}$. Consequently, a $C_0$ contraction semigroup $W_t$ on $L^2(\Omega)$ has its generator extended by $-i\hat{H}^*$ if and only if there exists a linear contraction $\Phi:L^2(\partial \Omega)\to L^2(\partial \Omega)$ such that $W_t=\exp(-it\hat{H}_\Phi)$ with
    \begin{equation}\label{Classification}
        D(\hat{H}_\Phi):=\{ \psi \in D(\hat{H}^*): \Phi G_- \psi= G_+ \psi\}, \quad \hat{H}_\Phi \psi:=(-\Delta+V)\psi.
    \end{equation}
\end{theoremmainagain} 
\begin{proof}
    Fix $\eta \in \rho(\hat{H}_D)\cap \mathbb{R}$. By Theorem \ref{Lumer-Phillips}, Theorem \ref{Description}, and equation (\ref{Current Star Decomp Eq}), it suffices to show
    \begin{equation}
        (G_+, G_-):D(\hat{H}^*)\to L^2(\partial \Omega)\times L^2(\partial \Omega)\text{ is surjective.}
    \end{equation}
    To that end, let $\xi, \chi \in L^2(\partial \Omega)$ and consider $\iota_-^{-1}\chi \in H^{-1/2}(\partial \Omega)$, $\iota_+^{-1}\xi \in H^{1/2}(\partial \Omega)$. From (\ref{Trace Definition}) we have that $\tau_N$ is a surjective mapping from $D(\hat{H}_D)$ onto $H^{1/2}(\partial \Omega)$, hence there exists some $\phi_D \in D(\hat{H}_D)$ such that $\tau_N \phi_D=\iota_+^{-1}\xi$. Recall also that $\tilde{\tau}_D$ is surjective from $D(\hat{H}^*)$ to $H^{-1/2}(\partial \Omega)$, and that $\ker \tilde{\tau}_D=\ker \tau_D= D(\hat{H}_D)$. It follows from the direct sum decomposition (\ref{Dirichlet Star Decomp}) that the restriction $\tilde{\tau}_D: \ker(\hat{H}^*-\eta)\to H^{-1/2}(\partial \Omega)$ is bijective, hence there exists $\phi_\eta \in \ker(\hat{H}^*-\eta)$ such that $\tilde{\tau}_D \phi_\eta=\iota_-^{-1}\chi$. Consequently $\phi:=\phi_D + \phi_\eta \in D(\hat{H}^*)$ satisfies 
    \begin{equation}
        \iota_-\tilde{\tau}_D\phi=\iota_- \tilde{\tau}_D \phi_\eta=\iota_- \iota_-^{-1}\chi=\chi
    \end{equation}
    \begin{equation}
        \iota_+ \tau_N \phi_D=\iota_+ \iota_+^{-1}\xi=\xi. 
    \end{equation}
    It follows that 
    \begin{equation}
        \begin{pmatrix}
            \sqrt{2}(G_- + G_+)
            \\
            \sqrt{2}i (G_- - G_+)
        \end{pmatrix} :{D(\hat{H}^*}) \to L^2(\partial \Omega)\times L^2(\partial \Omega) \text{ is surjective}.
    \end{equation}
    The surjectivity of $(G_+, G_-)$ also follows immediately.
\end{proof}
\section{Generalized Robin Boundary Conditions}
In this section we fix a boundary quadruple and construct linear contractions $\Phi$ which correspond to commonly known boundary conditions for the Schr\"odinger operator. We will find it most convenient to fix $\eta \in \rho(\hat{H}_D)\cap \rho(\hat{H}_N)\cap \mathbb{R}$, which exists since $\hat{H}_D - V$ and $\hat{H}_N - V$ are positive operators and $V \in L^\infty(\Omega,\mathbb{R})$. Any choice of $\eta \in \rho(\hat{H}_D)\cap \rho(\hat{H}_N) \cap \mathbb{R}$ can be used to prove the results in this section, although we will eventually find it useful to set $\eta<-(1+||V||_{L^\infty(\Omega)})$. For a given contraction $\Phi$, the $m$-dissipative extension $-i\hat{H}_\Phi$ of $-i\hat{H}$ is defined via
\begin{equation}\label{H Phi Def}
        D(\hat{H}_\Phi):=\left\{\psi \in \hat{H}^*: G_+ \psi= \Phi G_-\psi \right\}= \left\{\psi \in \hat{H}^*: i(\mathds{1}-\Phi)\iota_-\tilde{\tau}_D \psi=(\mathds{1}+\Phi)\iota_+ \tau_N \psi_D \right\}.
    \end{equation}
\subsection{Dirichlet and Neumann boundary conditions}
\begin{example}\textbf{Dirichlet} boundary conditions: \em{Setting $\Phi=-\mathds{1}$ returns
    \begin{equation}
        D(\hat{H}_{-\mathds{1}})= \left\{\psi \in \hat{H}^*: \tilde{\tau}_D \psi=0 \right\}.
    \end{equation}
    From Lemma \ref{Trace Extensions} we know $\ker(\tilde{\tau}_D)=\ker(\tau_D)=D(\hat{H}_D)$, so $\hat{H}_{-\mathds{1}}= \hat{H}_D$.}
\end{example}
\begin{example}\label{Krein Extension}
    \textbf{Krein type extension}: \em{Setting $\Phi=\mathds{1}$ returns
    \begin{equation}
        D(\hat{H}_{\mathds{1}})= \left\{\psi \in \hat{H}^*: \tau_N \psi_D=0 \right\}.
    \end{equation}
This must not be mistaken for the Neumann extension, as $\psi_D=(\hat{H}_D-\eta)^{-1}(\hat{H}^*-\eta)\psi \neq \psi$.}
\end{example}
To recover the Neumann boundary condition we must construct a contraction $\Phi$ so that $\tilde{\tau}_D \psi$ and $\tau_N \psi_D$ drop out of the boundary condition, and we are just left with $\tilde{\tau}_N \psi=0$. To accomplish this, we introduce the ``Dirichlet solution" map $\gamma(\lambda)$ and the ``Dirichlet-to-Neumann" map $D(\lambda)$.
\begin{definition}
    For $\lambda \in \rho(\hat{H}_D)\cap\mathbb{R}$, the ``Dirichlet solution" map $\gamma(\lambda)$ is defined by
    \begin{equation}
        \gamma(\lambda):\xi \mapsto \psi_\lambda \quad \text{ where} \quad \begin{cases}
                (\hat{H}^*-\lambda)\psi_\lambda=0
                \\
                \psi_\lambda \big{|}_{\partial \Omega}=\xi
            \end{cases}
    \end{equation}
    The ``Dirichlet-to-Neumann" map is defined as $D(\lambda):=\tau_N\gamma(\lambda)$, it maps Dirichlet values to Neumann values $D(\lambda):\tau_D\psi_\lambda \mapsto\tau_N\psi_\lambda$. 
    \end{definition}
    These mappings are well-studied in the literature and are famously used in the study of inverse problems. We will primarily refer the reader to section \ref{D to N section} of the Appendix for the proofs of Lemmas regarding these mappings, such as the one below.
    \begin{lemma}\label{maps properties}
        For $\lambda \in \rho(\hat{H}_D)\cap \mathbb{R}$, the map $\gamma(\lambda)$ is a bounded linear operator from $H^{3/2}(\partial \Omega)\to H^2(\Omega)$. Consequently, the ``Dirichlet-to-Neumann" map $D(\lambda)$ is a bounded linear map from $H^{3/2}(\partial \Omega)\to H^{1/2}(\partial \Omega)$. $D(\lambda)$ is also a densely defined symmetric operator on $L^2(\partial \Omega)$, and for $\lambda \in \rho(\hat{H}_D)\cap \rho(\hat{H}_N)\cap\mathbb{R}$ the operator $D(\lambda):H^{3/2}(\partial \Omega)\to H^{1/2}(\partial \Omega)$ is bijective with bounded inverse.
    \end{lemma}
\begin{proof}
    See the proofs of Lemma \ref{Dirichlet solution properties} and Lemma \ref{Dirichlet to Neumann Facts} in the Appendix.
\end{proof}
We also require an additional Lemma.
\begin{lemma}
    Let $\lambda \in \rho(\hat{H}_D)\cap \rho(\hat{H}_N)\cap \mathbb{R}$. Then the operator $\Theta_N(\lambda) :H^2(\partial \Omega)\to L^2(\partial \Omega)$, defined as $\Theta_N(\lambda):=\iota_+ D(\lambda)\iota_-^{-1}$ is a densely defined bijective symmetric operator in $L^2(\partial \Omega)$, and is hence an unbounded self-adjoint operator on $L^2(\partial \Omega)$ with $\pm i \in \rho(\Theta_N)$.
\end{lemma}
\begin{proof}
    It suffices to show that $\langle \Theta_N(\lambda) \chi, \chi \rangle_{L^2(\partial \Omega)}\in \mathbb{R}$ for all $\chi \in H^2(\partial \Omega)$. Let $\chi$ be such an element, and let $\xi=\iota_-^{-1}\chi \in H^{3/2}(\partial \Omega)$. Then
    \begin{equation}
        \langle \iota_+ D(\lambda)\iota_-^{-1}\chi,\chi \rangle_{L^2(\partial \Omega)}=\langle \iota_+ D(\lambda)\chi,\iota_-\chi \rangle_{L^2(\partial \Omega)}=\langle  D(\lambda)\xi,\xi \rangle_{L^2(\partial \Omega)}.
    \end{equation}
    Hence $\langle \Theta_N \chi,\chi \rangle_{L^2(\partial \Omega)}=\langle  D(\lambda)\xi,\xi \rangle_{L^2(\partial \Omega)}\in \mathbb{R}$ for all $\chi \in H^2(\partial \Omega)$ as desired.
\end{proof}
\begin{example}
    \textbf{Neumann} boundary condition: \em{Set $\Theta_N=\Theta_N(\eta)$ and $\Phi=(i + \Theta_N)(i- \Theta_N)^{-1}$, which is well defined since $i \in \rho(\Theta_N)$. It is easy to verify
    \begin{equation}\label{Phi Simplication}
        (\mathds{1}+\Phi)=2i (i-\Theta_N)^{-1}, \quad (\mathds{1}-\Phi)=-2\Theta_N (i-\Theta_N)^{-1}.
    \end{equation}
    so the domain of $\hat{H}_\Phi$ is given by
    \begin{equation}
        D(\hat{H}_{\Phi})= \left\{\psi \in \hat{H}^*: -\Theta_N (i- \Theta_N)^{-1}\iota_-\tilde{\tau}_D \psi=(i - \Theta_N)^{-1}\iota_+ \tau_N \psi_D \right\}.
    \end{equation}
    For $\psi \in D(\hat{H}_\Phi)$ our boundary condition implies
    \begin{equation}
    \iota_- \tilde{\tau}_D \psi=(i-\Theta_N)(i-\Theta_N)^{-1}\iota_-\tilde{\tau}_D\psi=i(i-\Theta_N)^{-1}\iota_-\tilde{\tau}_D\psi +(i-\Theta_N)^{-1}\iota_+\tau_N \psi_D\end{equation}
    hence $\iota_- \tilde{\tau}_D \psi \in D(\Theta_N)=H^2(\partial \Omega)$. Consequently, the trace  $\tilde{\tau}_D \psi \in H^{3/2}(\partial \Omega)$ which by Lemma \ref{Trace Extensions} implies that $\psi \in H^2(\Omega)$, so $\tilde{\tau}_D \psi= \tau_D \psi$.  Applying $(i- \Theta_N)$ to both sides of the boundary condition returns
    \begin{equation}
        \iota_+ \tau_N \psi_D= -\Theta_N \iota_- \tau_D \psi=-\Theta_N \iota_- \tau_D \psi_\eta=-\iota_+ D(\eta)\tau_D \psi_\eta=-\iota_+ \tau_N \psi_\eta=  \iota_+ (\tau_N \psi_D-\tau_N \psi)
    \end{equation}
    so $\tau_N \psi=0$, and we have $\hat{H}_\Phi= \hat{H}_N$ as desired.}
 \end{example}
 \subsection{Regular Robin boundary conditions}
 In what follows we will be primarily interested in extensions of $-i\hat{H}$ that generate non-unitary dynamics. The extensions will have compact resolvents, so the following lemma will be useful in proofs concerning asymptotic stability.
 \begin{lemma}\label{no eigenvalues}
        For $\Omega \subset \mathbb{R}^n$ a bounded $C^2$ domain, the minimal operator $\hat{H}=(-\Delta+V):H^2_0(\Omega)\to L^2(\Omega)$ admits no eigenvalues.
    \end{lemma}
    \begin{proof}
        We proceed by contradiction. Suppose that there exists a $\lambda \in \mathbb{C}$ and $\psi_\lambda \in H^2_0(\Omega)$ such that $\hat{H}\psi_\lambda=\lambda \psi_\lambda$. Denote $V_+:=V+||V||_{L^\infty(\Omega)}$, so $\psi_\lambda$ is an eigenvector of $\left(-\Delta + V_+\right)$ with eigenvalue $\lambda + ||V||_{L^\infty(\Omega)}$. We may extend $\psi_\lambda$ to a function $\overline{\psi}_\lambda$ on $\mathbb{R}^n$ by setting $\overline\psi_\lambda=0$ on $\mathbb{R}^n\setminus \Omega$. Since this extension mapping $\psi \mapsto \overline{\psi}$ is a continuous linear operator \cite{Brezis} from $H^2_0(\Omega)\to H^2(\mathbb{R}^n)$, it follows that $\overline{\psi}_\lambda$ is an eigenvector of the operator $(-\Delta+\overline{V}_+):H^2(\mathbb{R}^n)\to L^2(\mathbb{R}^n)$ with eigenvalue $\lambda+||V||_{L^\infty(\Omega)}$. However, it can be shown that this operator has no eigenvalues. To see this we note that $(-\Delta+\overline{V}_+)$ is a positive operator on $L^2(\Omega)$ so $\sigma(\Delta+\overline V_+)\subset [0,\infty)$. By Rellich–Kondrachov \cite{Brezis} $\mathds{1}_{\Omega}\circ(-\Delta+\mu)^{-1}:L^2(\mathbb{R}^n)\to H^2(\mathbb{R}^n)\to H^2(\Omega)\hookrightarrow L^2(\Omega)\to L^2(\mathbb{R}^n) $ is a compact operator on $L^2(\mathbb{R}^n)$ for any $\mu \in \rho(-\Delta)$, so $-\Delta + \overline{V}_+$ is a relatively compact perturbation of $-\Delta$. This implies that the essential spectrum of these operators are the same \cite[Theorem 8.4.3]{Davies_1995}, $\sigma_\text{ess}(-\Delta + \overline{V}_+)=\sigma_\text{ess}(-\Delta)=[0,\infty)$. Hence  $\sigma(-\Delta + \overline{V}_+)=\sigma_\text{ess}(-\Delta+\overline V_+)$, contradicting our early statement that $\overline{\psi}_\lambda$ is an eigenvector of $-\Delta + \overline{V}_+$. 
    \end{proof}
    \begin{theorem}\label{Robin boundary conditions nice}
    (regular \textbf{Robin} boundary condition) Let $\beta: H^{3/2}(\partial \Omega)\to H^{1/2}(\partial \Omega)$ be a compact operator such that \em$\text{Re} \langle \beta \chi, \chi \rangle_{L^2(\partial \Omega)}\geq 0$ \em for all $\chi \in H^{3/2}(\partial \Omega)$. Then the operator defined via
    \begin{equation}
        D(-i\hat{H}_\beta):=\left\{\psi \in H^2(\Omega):\tau_N \psi=i \beta \tau_D \psi \right\}, \quad -i\hat{H}_\beta:=-i\hat{H}^*\big{|}_{D(-i\hat{H}_\beta)}
    \end{equation}
    is an $m$-dissipative extension of $-i\hat{H}$. If the real part of $\beta$ is strictly positive, \em $\text{Re}\langle \beta \chi, \chi \rangle_{L^2(\partial \Omega)}>0$ for all $0 \neq \chi \in H^{3/2}(\partial \Omega)$\em, then $||\exp(-it\hat{H}_\beta)\psi_0||_{L^2(\Omega)}\xrightarrow{t \to \infty}0$ for each $\psi_0 \in L^2(\Omega)$.
\end{theorem}
As before, our goal is to construct a contraction $\Phi$ so that $\tau_N \psi_D$ drops out of the boundary condition and we are left with a relation between $\tau_D\psi$ and $\tau_N \psi$. Towards this goal we first prove the following lemma.
    \begin{lemma}
        Let $\lambda \in \rho(\hat{H}_D)\cap \rho(\hat{H}_N)\cap \mathbb{R}$. Then the operator $\Theta_\beta(\lambda):=\Theta_N(\lambda)- i\iota_+ \beta \iota_-^{-1}:H^2(\partial \Omega)\to L^2(\partial \Omega)$ is a densely defined operator in $L^2(\partial \Omega)$, with $i \in \rho(\Theta_\beta(\lambda))$.
    \end{lemma}
    \begin{proof}
        Recall that $i - \Theta_N(\lambda):H^2(\partial \Omega)\to L^2(\partial \Omega)$ is a bijective and bounded linear map, and is thus Fredholm with index $0$. Since $\beta$ is compact and $\iota_-^{-1}:H^2(\partial \Omega) \to H^{3/2}(\partial \Omega)$, $\iota_+:H^{1/2}(\partial \Omega)\to L^2(\partial \Omega)$ are isometries, it follows that $\iota_+ \beta \iota_-^{-1}:H^2(\partial \Omega)\to L^2(\partial \Omega)$ is also compact. By the stability of Fredholm operators under compact perturbations \cite{Brezis}, we have that $i-\Theta_\beta(\lambda)=i-\Theta_N(\lambda)+ i\iota_+ \beta \iota_-^{-1}$ is also Fredholm of index $0$, so it suffices to prove that the map is injective. Let $\xi \in H^2(\partial \Omega)$ with $\chi=\iota_-^{-1}\xi \in H^{3/2}(\partial \Omega)$. Since $\Theta_N(\lambda)$ is symmetric
        \begin{equation}
            \text{Im} \langle (i - \Theta_\beta(\lambda))\xi, \xi \rangle_{L^2(\partial \Omega)}= \text{Im} \langle i\xi,\xi \rangle_{L^2(\partial \Omega)} + \text{Re} \langle \iota_+ \beta \iota_-^{-1}\xi,\xi \rangle_{L^2(\partial \Omega)}= ||\xi||^2_{L^2(\partial \Omega)} + \text{Re} \langle \beta \chi,\chi \rangle_{L^2(\partial \Omega)}.
        \end{equation}
        Hence 
        \begin{equation}
            ||\xi||_{L^2(\partial \Omega)}^2 \leq | \text{Im} \langle (i - \Theta_\beta(\lambda))\xi, \xi \rangle_{L^2(\partial \Omega)} | \leq ||(i - \Theta_\beta(\lambda)) \xi||_{L^2(\partial \Omega)} ||\xi||_{L^2(\partial \Omega)}
        \end{equation}
        which implies $||\xi||_{L^2(\partial \Omega)}\leq ||(i - \Theta_\beta(\lambda)) \xi ||_{L^2(\partial \Omega)}$. It follows that $(i-\Theta_\beta(\lambda))$ is an injective Fredholm operator of index $0$ and is thus invertible.
    \end{proof}
    \begin{proof}[Proof of Theorem \ref{Robin boundary conditions nice}]
        Set $\Theta_\beta=\Theta_\beta(\eta)$ and $\Phi=(i + \Theta_\beta)(i- \Theta_\beta)^{-1}$. It is easy to verify
    \begin{equation}
        (\mathds{1}+\Phi)=2i (i-\Theta_\beta)^{-1}, \quad (\mathds{1}-\Phi)=-2\Theta_\beta (i-\Theta_\beta)^{-1}.
    \end{equation}
    The domain of $\hat{H}_\Phi$ is then given by
    \begin{equation}
        D(\hat{H}_\Phi)= \left\{\psi \in D(\hat{H}^*): -\Theta_\beta (i- \Theta_\beta)^{-1}\iota_-\tilde{\tau}_D \psi=(i - \Theta_\beta)^{-1}\iota_+ \tilde{\tau}_N \psi_D \right\}.
    \end{equation}
    To show $\hat{H}_\Phi$ is a maximally dissipative extension of $\hat{H}$, we must prove $||(i+\Theta_\beta)(i-\Theta_\beta)^{-1}\chi||_{L^2(\partial \Omega)}\leq ||\chi||_{L^2(\partial \Omega)}$. Setting $\xi=(i-\Theta_\beta)^{-1}\chi$, it is equivalent to prove $||(i+\Theta_\beta)\xi||_{L^2(\partial \Omega)}\leq ||(i-\Theta_\beta)\xi||_{L^2(\partial \Omega)}$ for all $\xi \in H^{2}(\partial \Omega)$. This follows immediately from
    \begin{equation}\label{contraction regular}
        ||(i-\Theta_\beta)\xi||_{L^2(\partial \Omega)}^2-||(i+\Theta_\beta)\xi||_{L^2(\partial \Omega)}^2=4\text{Re}\langle \iota_+\beta\iota_-^{-1}\xi,\xi\rangle_{L^2(\partial \Omega)}=4\text{Re}\langle \beta\iota_-^{-1}\xi, \iota_-^{-1}\xi\rangle_{L^2(\partial \Omega)}\geq 0.
    \end{equation}
    Repeating the same steps as in the Neumann example, the boundary condition implies $\iota_- \tilde{\tau}_D \psi \in D(\Theta_\beta)=H^2(\partial \Omega)$ for all $\psi \in D(\hat{H}_\Phi)$, and in particular $\tilde{\tau}_D \psi \in H^{3/2}(\partial \Omega)$ which by Lemma \ref{Trace Extensions} implies that $\psi \in H^2(\Omega)$, so $\tilde{\tau}_D \psi= \tau_D \psi$.  Applying $(i- \Theta_\beta)$ to both sides of the boundary condition returns
    \begin{equation}
    \begin{split}
         \iota_+ \tau_N \psi_D&= - (\Theta_N- i\iota_+ \beta \iota_-^{-1}) \iota_- \tau_D \psi
        \\
        &= (i\iota_+ \beta \tau_D \psi-\Theta_N\iota_- \tau_D \psi_\eta) 
        \\
        &=\iota_+( i \beta \tau_D \psi-D(\eta) \tau_D \psi_\eta) 
        \\
        &=\iota_+(i \beta  \tau_D \psi-\tau_N \psi + \tau_N \psi_D ) 
    \end{split}
    \end{equation}
    so $D(\hat{H}_\Phi)=  \left\{\psi \in H^2(\Omega):\tau_N \psi=i \beta \tau_D \psi \right\} $ as desired. 
    \par
    Now, when the real part of $\beta$ is \textit{strictly} positive we have from equation (\ref{contraction regular}) that $\Phi$ is \textit{strictly} contractive on $L^2(\Omega)$. The contraction semigroup generator $-i\hat{H}_\Phi$ has compact resolvent because for $\lambda \in \rho(-i\hat{H}_\Phi)$ the resolvent operator $(-i\hat{H}_\Phi - \lambda)^{-1}$ maps $L^2(\Omega)$ to $H^2(\Omega)$, which by Rellich-Kondrachov \cite{Brezis} embeds compactly into $L^2(\Omega)$. The generator $-i\hat{H}_\Phi$ meets the requirements of Theorem \ref{asympt stability}, so $||\exp(-it\hat{H}_\Phi)\psi_0||_{L^2(\Omega)}\xrightarrow{t\to \infty}0$ holds for all $\psi_0 \in L^2(\Omega)$ if the operator $(-\Delta+V):H^2_0(\Omega)\to L^2(\Omega)$ admits no real eigenvalues. This follows immediately from Lemma \ref{no eigenvalues}, concluding our proof of Theorem \ref{Robin boundary conditions nice}.
 \end{proof}
 \begin{remark}
        Since the embeddings $H^{1/2 +\epsilon}(\partial \Omega) \xhookrightarrow{} H^{1/2}(\partial \Omega)$ and $H^{3/2}(\partial \Omega) \xhookrightarrow{} H^{3/2-\epsilon}(\partial \Omega)$ are compact \cite{Brezis} for any $\epsilon>0$, a sufficient condition for $\beta$ to be compact is taking $\beta:H^{3/2}(\partial \Omega)\to H^{1/2+\epsilon}(\partial \Omega)$ or $\beta:H^{3/2 -\epsilon}(\partial \Omega)\to H^{1/2}(\partial \Omega)$ bounded for some $\epsilon>0$. If $\beta$ is multiplication by some function, it suffices to take $\beta \in H^s(\partial \Omega)$ for any $s>\max \{\frac{1}{2},\frac{n-2}{2}\}$ \cite{behzadan2021multiplicationsobolevspacesrevisited}.
    \end{remark}
    \subsection{Generalized Robin boundary conditions with compactness}
    The construction provided in the previous subsection does not allow $\beta$ to be multiplication by some function of low regularity, such as a generic $L^\infty$ function. Our generalization to lower regularity $\beta$ will be slightly more complicated, and we no longer expect $D(\hat{H}_\Phi)$ to reside within $H^2(\Omega)$ since $\tilde \tau_N\psi=i\beta\tilde \tau_D\psi \notin H^{1/2}(\partial \Omega)$.
    \begin{theorem}\label{Robin boundary conditions ugly}
    (\textbf{Robin} boundary condition) Let $\beta: H^{1/2}(\partial \Omega)\to H^{-1/2}(\partial \Omega)$ be a compact operator such that \em$\text{Re} \langle \beta \chi, \chi \rangle_{H^{-1/2}(\partial \Omega)\times H^{1/2}(\partial \Omega)}\geq 0$ \em for all $\chi \in H^{1/2}(\partial \Omega)$. Then operator defined via
    \begin{equation}
        D(-i\hat{H}_\beta):=\left\{\psi \in H^1(\Omega)\cap D(\hat{H}^*):\tilde{\tau}_N \psi=i \beta \tilde{\tau}_D \psi \right\}, \quad -i\hat{H}_\beta:=-i\hat{H}^*\big{|}_{D(-i\hat{H}_\beta)}
    \end{equation}
    is a maximally dissipative extension of $-i\hat{H}$.  If the real part of $\beta$ is strictly positive;\newline \em $\text{Re}\langle \beta \chi, \chi \rangle_{H^{-1/2}(\partial \Omega)\times H^{1/2}(\partial \Omega)}>0$ for all $0 \neq \chi \in H^{1/2}(\partial \Omega)$\em; then $||\exp(-it\hat{H}_\beta)\psi_0||_{L^2(\Omega)}\xrightarrow{t \to \infty}0$ for each $\psi_0 \in L^2(\Omega)$. 
\end{theorem} The proof of this theorem will require us to introduce the ``extended" Dirichlet-to-Neumann maps.
\begin{lemma}\label{Extended maps properties}
    For $\lambda \in \rho(\hat{H}_D)\cap \mathbb{R}$, the operator $\gamma(\lambda)$ admits extensions to bounded linear operators $\gamma(\lambda):H^s(\partial \Omega)\to H^{s+1/2}(\Omega)$ for all $s \in [-\frac{1}{2},\frac{3}{2}]$. Consequently, the ``Dirichlet-to-Neumann" map $D(\lambda)$ also admits extensions to bounded linear maps $\tilde{D}(\lambda):=\tilde{\tau}_N\gamma(\lambda):H^{s}(\partial \Omega)\to H^{s-1/2}(\partial \Omega)$ for all $s \in [-\frac{1}{2},\frac{3}{2}]$. 
    \par
    The restriction $\tilde{D}(\lambda):H^{1/2}(\partial \Omega)\to H^{1/2}(\partial \Omega)$ is a ``symmetric" operator, in the sense that $\langle \tilde{D}(\lambda)\zeta,\zeta \rangle_{H^{-1/2}(\partial \Omega)\times H^{1/2}(\partial \Omega)}\in\mathbb{R}$ for all $\zeta \in H^{1/2}(\partial \Omega)$. Lastly, for $\lambda \in \rho(\hat{H}_D)\cap\rho(\hat{H}_N)\cap \mathbb{R}$ the restrictions $\tilde{D}(\lambda):H^s(\partial \Omega)\to H^{s-1/2}(\partial \Omega)$ are bijective linear maps with bounded inverses for all $s\in [-\frac{1}{2},\frac{3}{2}]$.
\end{lemma}
\begin{proof}
    See the proofs of Lemma \ref{Dirichlet solution properties} and Lemma \ref{Dirichlet to Neumann Facts} in the Appendix.
\end{proof}
\begin{corollary}
    For $\lambda \in \rho(\hat{H}_D)\cap \mathbb{R}$, the operator $\Theta_N(\lambda)=\iota_+\tilde{D}(\lambda)\iota_-^{-1}:H^{1}(\partial \Omega)\to H^{-1}(\partial \Omega)$ is ``symmetric" in the sense that $\langle \tilde{\Theta}_N(\lambda)\xi, \xi \rangle_{H^{-1}(\partial \Omega)\times H^1(\partial \Omega)} \in \mathbb{R}$ for all $\xi \in H^1(\partial \Omega)$. For $\lambda \in \rho(\hat{H}_D)\cap \rho(\hat{H}_N)\cap \mathbb{R}$ this operator is also bijective.
\end{corollary}
    \begin{lemma}
        Let $\lambda \in \rho(\hat{H}_D)\cap \rho(\hat{H}_N)\cap \mathbb{R}$. Then the operator $(i-\tilde\Theta_\beta(\lambda)):=(i-\tilde{\Theta}_N(\lambda)+ i\iota_+ \beta \iota_-^{-1}):H^1(\partial \Omega)\to H^{-1}(\partial \Omega)$ is a bijective linear operator with bounded inverse.
    \end{lemma}
    \begin{proof}
        Recall that $\tilde{\Theta}_N(\lambda):H^1(\partial \Omega)\to H^{-1}(\partial \Omega)$ is a bijective and bounded linear map, and is thus Fredholm with index $0$. Since $\beta$ is compact and $\iota_-^{-1}:H^{1}(\partial \Omega) \to H^{1/2}(\partial \Omega)$, $\iota_+:H^{-1/2}(\partial \Omega)\to H^{-1}(\partial \Omega)$ are isometries, it follows that $\iota_+ \beta \iota_-^{-1}:H^{1}(\partial \Omega)\to H^{-1}(\partial \Omega)$ is also compact. By compact embedding of Sobolev spaces, multiplication by $i$ is also a compact operator from $H^{1}(\partial \Omega)\to H^{-1}(\partial \Omega)$. So, the stability of Fredholm operators under compact perturbations \cite{Brezis} informs us that $i-\tilde{\Theta}_\beta(\lambda)=i-\tilde{\Theta}_N(\lambda)+ i\iota_+ \beta \iota_-^{-1}$ is also Fredholm of index $0$, and it again suffices to prove that the map is injective. Let $\xi \in H^1(\partial \Omega)$ with $\chi=\iota_-^{-1}\xi \in H^{1/2}(\partial \Omega)$. Since $\tilde{\Theta}_N(\lambda)$ is symmetric
        \begin{equation}\label{injectivity inequality}
            \text{Im} \langle (i - \tilde{\Theta}_\beta(
            \lambda
            ))\xi, \xi \rangle_{H^{-1}(\partial \Omega)\times H^{1}(\partial \Omega)}= ||\xi||^2_{L^2(\partial \Omega)} + \text{Re} \langle \beta \chi,\chi \rangle_{H^{-1/2}(\partial \Omega)\times H^{1/2}(\partial \Omega)}\geq ||\xi||_{L^2(\partial \Omega)}^2.
        \end{equation}
        Consequently $\xi=0$ whenever $(i-\tilde{\Theta}_\beta(\lambda))\xi=0$, so $(i-\tilde\Theta_\beta(
        \lambda
        ))$ is an injective Fredholm operator of index $0$ and is thus invertible.
    \end{proof}
    \begin{proof}[Proof of Theorem \ref{Robin boundary conditions ugly}]
        Set $\tilde{\Theta}_\beta=\tilde{\Theta}_\beta(\eta)$ and $\Phi=(i + \tilde{\Theta}_\beta)(i-\tilde{\Theta}_\beta)^{-1}\big{|}_{L^2(\partial \Omega)}$. Since $(i+\tilde{\Theta}_\beta)$ maps its domain into $H^{-1}(\partial \Omega)$, it is not immediately clear that $\Phi$ maps into $L^2(\partial \Omega)$. Before we continue, we will prove that $||(i+\tilde{\Theta}_\beta)(i-\tilde{\Theta}_\beta)^{-1}\chi||_{L^2(\partial \Omega)}\leq ||\chi||_{L^2(\partial \Omega)}$ for all $\chi \in L^{2}(\partial \Omega)$. Setting $\xi=(i-\tilde{\Theta}_\beta)^{-1}\chi\in H^1(\partial \Omega)$, it is equivalent to prove $||(i+\tilde{\Theta}_\beta)\xi||_{L^2(\partial \Omega)}\leq ||(i-\tilde{\Theta}_\beta)\xi||_{L^2(\partial \Omega)}$. This follows immediately from
    \begin{equation}\label{contraction}
    \begin{split}
        ||(i-\tilde{\Theta}_\beta)\xi||_{L^2(\partial \Omega)}^2-||(i+\tilde{\Theta}_\beta)\xi||_{L^2(\partial \Omega)}^2&=4\text{Im}\langle \tilde{\Theta}_\beta \xi, \xi \rangle_{H^{-1}(\partial \Omega)\times H^{1}(\partial \Omega)}
        \\
        &=4\text{Re}\langle \beta\iota_-^{-1}\xi, \iota_-^{-1}\xi\rangle_{H^{-1/2}(\partial \Omega)\times H^{1/2}(\partial \Omega)}\geq 0.
    \end{split}
    \end{equation}
    $\Phi$ is therefore contractive on $L^2(\partial \Omega)$, and the associated extension $-i\hat{H}_\Phi$ is maximally dissipative with domain
    \begin{equation}\label{H Phi bad def}
        D(\hat{H}_\Phi)= \left\{\psi \in D(\hat{H}^*): -\tilde{\Theta}_\beta (i- \tilde{\Theta}_\beta)^{-1}\iota_-\tilde{\tau}_D \psi=(i - \tilde{\Theta}_\beta)^{-1}\iota_+ \tilde{\tau}_N \psi_D \right\}.
    \end{equation}
    For $\psi \in D(\hat{H}_\Phi)$ our boundary condition implies
    \begin{equation}
    \iota_- \tilde{\tau}_D \psi=(i-\tilde{\Theta}_\beta)(i-\tilde{\Theta}_\beta)^{-1}\iota_-\tilde{\tau}_D\psi=i(i-\tilde{\Theta}_\beta)^{-1}\iota_-\tilde{\tau}_D\psi +(i-\tilde{\Theta}_\beta)^{-1}\iota_+\tau_N \psi_D
    \end{equation}
    so $\iota_- \tilde{\tau}_D \psi \in D(\tilde{\Theta}_\beta)=H^1(\partial \Omega)$, and in particular $\tilde{\tau}_D \psi \in H^{1/2}(\partial \Omega)$. By Corollary \ref{Converse Trace Regularity} it follows that $\psi \in H^1(\Omega)$, so $D(\hat{H}_\Phi)\subset H^1(\Omega)$.
    
    Applying $(i- \Theta_\beta)$ to both sides of the boundary condition in (\ref{H Phi bad def}) returns
    \begin{equation}
    \begin{split}
         \iota_+ \tilde\tau_N \psi_D&= - (\tilde\Theta_N- i\iota_+ \beta \iota_-^{-1}) \iota_- \tilde\tau_D \psi
        \\
        &= (i\iota_+ \beta  \tilde\tau_D \psi-\tilde\Theta_N\iota_- \tilde\tau_D \psi_\eta) 
        \\
        &=\iota_+( i \beta \tilde\tau_D \psi-\tilde{D}(\eta) \tilde\tau_D \psi_\eta) 
        \\
        &=\iota_+(i \beta  \tilde\tau_D \psi-\tilde\tau_N \psi + \tilde\tau_N \psi_D ) 
    \end{split}
    \end{equation}
    so $D(\hat{H}_\Phi)=  \left\{\psi \in H^1(\Omega)\cap D(\hat{H}^*):\tilde\tau_N \psi=i \beta \tilde\tau_D \psi \right\}$ as desired.
    
    When the real part of $\beta$ is \textit{strictly} positive we repeat the last steps in the proof of Theorem \ref{Robin boundary conditions nice}. $\Phi$ is \textit{strictly} contractive on $L^2(\Omega)$ by equation (\ref{contraction}), and  $-i\hat{H}_\Phi$ has compact resolvent because for $\lambda \in \rho(-i\hat{H}_\Phi)$ the resolvent operator $(-i\hat{H}_\Phi - \lambda)^{-1}$ maps $L^2(\Omega)$ to $H^1(\Omega)$, which by Rellich-Kondrachov \cite{Brezis} embeds compactly into $L^2(\Omega)$. The generator $-i\hat{H}_\Phi$ meets the requirements of Theorem \ref{asympt stability}, so $||\exp(-it\hat{H}_\Phi)\psi_0||_{L^2(\Omega)}\xrightarrow{t\to \infty}0$ holds for all $\psi_0 \in L^2(\Omega)$ if the operator $(-\Delta+V):H^2_0(\Omega)\to L^2(\Omega)$ admits no real eigenvalues. This follows immediately from Lemma \ref{no eigenvalues}, and we conclude our proof of Theorem \ref{Robin boundary conditions ugly}.
 \end{proof}
 \begin{remark}
        Since the embeddings $H^{-1/2 +\epsilon}(\partial \Omega) \xhookrightarrow{} H^{-1/2}(\partial \Omega)$ and $H^{1/2}(\partial \Omega) \xhookrightarrow{} H^{1/2-\epsilon}(\partial \Omega)$ are compact \cite{Brezis} for any $\epsilon>0$, a sufficient condition for $\beta$ to be compact is taking $\beta:H^{1/2}(\partial \Omega)\to H^{-1/2+\epsilon}(\partial \Omega)$ or $\beta:H^{1/2 -\epsilon}(\partial \Omega)\to H^{-1/2}(\partial \Omega)$ bounded for some $\epsilon>0$.  If $\beta$ is multiplication by some function, it suffices to take $\beta \in L^\infty(\partial \Omega)$.
    \end{remark}
\subsection{Generalized Robin boundary conditions without compactness}\label{Robin uggliest section}
We turn our attention to the least regular class of Robin boundary operators $\beta:H^{1/2}(\partial \Omega)\to H^{-1/2}(\partial \Omega)$; those that are bounded with non-negative real part but are not compact operators. Looking back at our proof of Theorem \ref{Robin boundary conditions ugly}, we see that once the invertibility of $(i-\tilde\Theta_\beta(\eta))=(i - \iota_+(\tilde D(\eta)-i\beta)\iota_-^{-1}):H^{1}(\partial \Omega)\to H^{-1}(\partial \Omega)$ is established then all further steps are effectively identical to those in the proof of Theorem \ref{Robin boundary conditions nice}. 
However, the compactness assumption played a crucial role in our analysis: it allowed us to apply Fredholm theory to derive the surjectivity of $(i-\tilde\Theta_\beta(\eta))$ from its injectivity. We will demonstrate how this compactness assumption can be dropped from the statement of Theorem \ref{Robin boundary conditions ugly} at the expense of picking up alternative conditions. 
\par
We now restate and prove the second main result of this paper.
\begin{theoremmainagain}\label{Robin boundary conditions ugliest}
    (Generalized \textbf{Robin} boundary condition) Let $\Omega \subset \mathbb{R}^n$ be a bounded $C^2$ domain of dimension $n>1$, and let $\hat{H}:=(-\Delta+V)\big{|}_{C_c^\infty(\Omega)}$ with $V \in L^\infty(\Omega,\mathbb{R})$. Suppose that $\beta=\beta_1+\beta_2+\beta_3$ is a sum of three bounded linear operators $\beta_k:H^{1/2}(\partial \Omega)\to H^{-1/2}(\partial \Omega)$ having ``non-negative real parts", meaning \em $\text{Re}\langle \beta_k \chi,\chi \rangle_{H^{-1/2}(\partial \Omega)\times H^{1/2}(\partial \Omega)}\geq 0$ \em for all $\chi\in H^{1/2}(\partial \Omega)$, with the following conditions on $\beta_1$, $\beta_2$, and $\beta_3$:
    \begin{enumerate}
        \em \item[(R1)]\em Compactness: $\beta_1:H^{1/2}(\partial \Omega)\to H^{-1/2}(\partial \Omega)$ is a compact operator.
        \em \item[(R2)]\em Smallness in operator norm: $||\beta_2
        ||_{\mathcal{B}(H^{1/2}(\partial \Omega),H^{-1/2}(\partial \Omega))}< ||\tilde\tau_D||^{-2}_{\mathcal{B}(H^{1}( \Omega),H^{1/2}(\partial \Omega))}$.
        \em \item[(R3)]\em Non-negative imaginary part: \em $\text{Im}\langle \beta_3 \chi, \chi\rangle_{H^{-1/2}(\partial \Omega),H^{1/2}(\partial \Omega)}\geq 0$ \em for all $\chi \in H^{1/2}(\partial \Omega)$. 
    \end{enumerate}
    Then the operator defined by
    \begin{equation}\label{non compact Robin def}
        D(-i\hat{H}_\beta):=\left\{\psi \in H^1(\Omega)\cap D(\hat{H}^*):\tilde{\tau}_N \psi=i \beta \tilde{\tau}_D \psi \right\}, \quad -i\hat{H}_\beta:=-i\hat{H}^*\big{|}_{D(-i\hat{H}_\beta)}
    \end{equation}
    is a maximally dissipative extension of $-i\hat{H}$. Consequently, for each $\psi_0 \in D(-i\hat{H}_\beta)$ the initial-boundary value problem
    \em \begin{equation}\label{Ugliest Robin IBVP}
\left\{\begin{array}{rclr}
       i \partial_t \psi&=& \hat{H}^* \psi \quad &\text{in } \Omega
       \\
       \psi &=&\psi_0 \quad &\text{at } t=0
       \\
       \partial_n \psi&=&i \beta \psi \quad &\text{on } \partial \Omega
    \end{array}
    \right.
    \end{equation}
    \em admits a unique, global-in-time solution $\psi_t=\exp(-it\hat{H}_\beta)\in C^1([0,\infty),L^2(\Omega))$, and the solution mappings $W_t:\psi_0\mapsto \psi_t$ extend continuously to a $C_0$ contraction semigroup on $L^2(\Omega)$.
\end{theoremmainagain}
For any choice of $\eta \in \rho(\hat{H}_D)\cap\mathbb{R}$, Theorem \ref{Parameterization Schrod} offers an explicit parameterization of all maximally dissipative extensions of $-i\hat{H}$ in terms of linear contractions $\Phi$ on $L^2(\partial \Omega)$. In practice, certain choices of $\eta$ are more convenient when trying to construct the linear contraction associated with a particular extension. In the last three subsections we took $\eta \in \rho(\hat{H}_D)\cap \rho(\hat{H}_N)\cap \mathbb{R}$ when constructing the linear contractions associated with the Neumann and compact Robin extensions. For the proof of Theorem \ref{Robin boundary conditions ugliest} we will find it convenient to additionally demand $\eta < -(1+||V||_{L^\infty(\Omega)})$. Our construction of the linear contraction $\Phi$ associated with the extension $-i\hat{H}_\beta$ defined in (\ref{non compact Robin def}) relies entirely on the following Lemma.
\begin{lemma}\label{Theta beta invertibility}
    Let $\lambda <-(1+||V||_{L^\infty(\Omega)})$, so $\lambda \in \rho(\hat{H}_D)\cap \rho(\hat{H}_N)\cap \mathbb{R}$. Then the operator $(i - \tilde{\Theta}_\beta(\lambda)):=(i-\tilde{\Theta}_N(\lambda) + i \iota_+ \beta \iota_-^{-1}):H^{1}(\partial \Omega)\to H^{-1}(\partial \Omega)$ is a bijective linear operator with a bounded inverse.
\end{lemma}
\begin{proof}
    Repeating the same set of inequalities as in equation (\ref{injectivity inequality}) shows that the operator $(i - \tilde \Theta_\beta(\lambda))$ is injective. However, our proof for surjectivity will require more care. Since $\iota_+ \beta_1\iota_-^{-1}$ is a compact operator it does not impact the Fredholm index of $(i - \tilde \Theta_\beta(\lambda))$, so without loss of generality we may set $\beta_1=0$. Let $\xi \in H^1(\partial \Omega)$ and define $\chi=\iota_-^{-1}\xi \in H^{1/2}(\partial \Omega)$. Then by (R3)
    \begin{equation}\label{Real inequality}
        \text{Re}\langle ( \tilde{\Theta}_\beta(\lambda)-i )\xi,\xi \rangle_{H^{-1}(\partial \Omega)\times H^{1}(\partial \Omega)}\geq\langle \tilde D(\lambda) \chi, \chi\rangle_{H^{-1/2}(\partial \Omega)\times H^{1/2}(\partial \Omega)} + \text{Im}\langle \beta_2 \chi, \chi \rangle_{H^{-1/2}(\partial \Omega)\times H^{1/2}(\partial \Omega)}.
    \end{equation}
    Introducing $\psi_\lambda:=\tilde\gamma(\lambda)\chi\in H^1(\Omega)$, the unique element in $\ker(\hat{H}^*-\lambda)$ such that $\tilde \tau_D \psi_\lambda = \chi$, we may apply the Green's identity (\ref{Green's Identity Star}) and integration-by-parts to compute
    \begin{equation}
    \begin{split}
    \langle \tilde D(\lambda) \chi, \chi\rangle_{H^{-1/2}(\partial \Omega)\times H^{1/2}(\partial \Omega)}&=\langle\tilde{\tau}_N\psi_\lambda, \tilde{\tau}_D \psi_\lambda\rangle_{H^{-1/2}(\partial \Omega)\times H^{1/2}(\partial \Omega)}
    \\
    &=||\vec{\nabla} \psi_\lambda||_{L^2(\Omega)}^2 + \langle (V-\lambda)\psi_\lambda,\psi_\lambda\rangle_{L^2(\Omega)}
    \\
    & \geq  ||\vec{\nabla} \psi_\lambda||_{L^2(\Omega)}^2 + ||\psi_\lambda||^2_{L^2(\Omega)} = ||\psi_\lambda||^2_{H^1(\Omega)}
    \end{split}
    \end{equation}
    where we've used $\inf (V-\lambda)\geq1$. Let us recall the trace theorem for functions that reside in $H^1(\Omega)$.
    \begin{lemma}
        \cite[Theorem 1.5.1.2]{Grisvard} For $\Omega\subset \mathbb{R}^n$ a bounded $C^2$ domain, the trace map $\psi \mapsto \psi\big{|}_{\partial \Omega}$ defined for $\psi \in C^\infty(\Omega) \to H^{{1/2}}(\partial \Omega)$ admits a continuous extension $\tilde{\tau}_D:H^{1}(\Omega)\to H^{1/2}(\partial \Omega)$.
    \end{lemma}
    Denoting the norm of this bounded linear operator $C_D:=||\tilde{\tau}_D||_{\mathcal{B}(H^1(\Omega),H^{1/2}(\partial \Omega))}$, we see  
    \begin{equation}\label{D lambda bound}
       ||\xi||_{H^1(\partial \Omega)}^2=||\chi||^2_{H^{1/2}(\partial \Omega)}\leq C_D^2||\psi_\lambda||_{H^1(\Omega)}^2\leq C_D^2  \langle \tilde D(\lambda) \chi, \chi\rangle_{H^{-1/2}(\partial \Omega)\times H^{1/2}(\partial \Omega)}.
    \end{equation}
    Introducing the shorthand $||\beta_2||:=||\beta_2||_{\mathcal{B}(H^{1/2}(\partial \Omega),H^{-1/2}(\partial \Omega))}$, we recall from condition (R2) that $C_D^{-2}-||\beta_2||>0$. Also,
    \begin{equation}\label{Beta 2 bound}
        \left|\text{Im}\langle \beta_2 \chi, \chi \rangle_{H^{-1/2}(\partial \Omega)\times H^{1/2}(\partial \Omega)}\right|\leq ||\beta_2 \chi||_{H^{-1/2}(\partial \Omega)}||\chi||_{H^{1/2}(\partial \Omega)}\leq ||\beta_2||~||\chi||_{H^{1/2}(\partial \Omega)}^2=||\beta_2||~||\xi||_{H^{1}(\partial \Omega)}^2
    \end{equation}
    Combining the inequalities of equations (\ref{Real inequality}), (\ref{D lambda bound}), and (\ref{Beta 2 bound}), we arrive at
    \begin{equation}
     \begin{split}
    ||(i-\tilde{\Theta}_\beta(\lambda))\xi||_{H^{-1}(\partial \Omega)}||\xi||_{H^1(\partial \Omega)}&\geq \text{Re}\langle(\tilde{\Theta}_\beta(\lambda)-i)\xi,\xi \rangle_{H^{-1}(\partial \Omega)\times H^1(\partial \Omega)} 
    \\
    &\geq \langle \tilde D(\lambda)\chi, \chi\rangle_{H^{1/2}(\partial \Omega)\times H^{-1/2}(\partial \Omega)}-  ||\beta_2||~||\xi||_{H^1(\partial \Omega)}^2
    \\
    &\geq \left(C_D^{-2}-||\beta_2||\right)||\xi||_{H^1(\partial \Omega)}^2.
    \end{split}
    \end{equation}
    Hence $||\xi||_{H^1(\partial \Omega)}\leq \left(C_D^{-2}-||\beta_2||\right)^{-1} ||(i-\tilde \Theta _\beta(\lambda) )\xi||_{H^{-1}(\partial \Omega)}$ and
    \begin{equation}
        ||(i-\tilde \Theta_\beta(\lambda))^{-1}\zeta||_{H^1(\partial \Omega)}\leq \left(C_D^{-2}-||\beta_2||\right)^{-1}||\zeta||_{H^{-1}(\partial \Omega)}, \quad \text{for all }\zeta \in \text{Ran}(i-\tilde\Theta_\beta).
    \end{equation}
    This provides us with an operator norm bound of the inverse, but we have yet to prove that $(i-\tilde \Theta_\beta)$ is surjective. Towards that goal we first note that the range of $(i-\tilde \Theta_\beta)$ is closed, since for any sequence $\zeta_n=(i-\tilde \Theta_\beta(\lambda))\xi_n$ that is Cauchy in $H^{-1}(\partial \Omega)$
    \begin{equation}
        ||\xi_n-\xi_m||_{H^1(\partial \Omega)} \leq \left(C_D^{-2}-||\beta_2||\right)^{-1}||\zeta_n - \zeta_m||_{H^{-1}(\partial \Omega)}.
    \end{equation}
    So $\xi_n$ is also Cauchy in $H^{1}(\partial \Omega)$ with some limit $\xi\in H^1(\partial \Omega)$. The boundedness of $(i-\tilde \Theta_\beta(\lambda)):H^1(\partial \Omega)\to H^{-1}(\partial \Omega)$ ensures that $\zeta_n$ converges to $(i-\tilde \Theta_\beta(\lambda))\xi\in \text{Ran}(i-\tilde \Theta_\beta(\lambda))$ in $H^{-1}(\partial \Omega)$ norm.
    \par
    Now we may apply the closed range theorem to state that $\text{Ran}(i-\tilde\Theta_\beta(\lambda))=\ker(-i-\tilde\Theta_\beta(\lambda)^*)^\perp$, where $\tilde{\Theta}_\beta(\lambda)^*:H^{1}(\partial \Omega)\to H^{-1}(\partial \Omega)$ denotes the Banach space adjoint of $\tilde{\Theta}_\beta(\lambda)$. The inequalities we've already proven are sufficient to show that $\ker(-i-\tilde\Theta_\beta(\lambda)^*)=\{0\}$. Let $\xi \in H^1(\partial \Omega)$. Then by definition of the Banach space adjoint
    \begin{equation}
        -\text{Im}\langle(-i-\tilde\Theta_\beta(\lambda)^*)\xi,\xi\rangle_{H^{-1}(\partial \Omega)\times H^1(\partial \Omega)}=\text{Im}\langle(i-\tilde\Theta_\beta(\lambda))\xi,\xi\rangle_{H^{-1}(\partial \Omega)\times H^{1}(\partial \Omega)}\geq||\xi||_{L^2(\partial \Omega)}^2 
    \end{equation}
    so $\xi \in \ker(-i-\tilde\Theta_\beta(\lambda)^*)$ if and only if $\xi=0$. So $(i-\tilde \Theta_\beta(\lambda))$ is a bounded bijective operator from $H^{1}(\partial \Omega)$ onto $H^{-1}(\partial \Omega)$ as desired.
\end{proof}
\begin{proof}[Proof of Theorem \ref{Robin boundary conditions ugliest}] Set $\tilde{\Theta}_\beta=\tilde{\Theta}_\beta(\eta)$ and $\Phi=(i+\tilde{\Theta}_\beta)(i-\tilde{\Theta}_\beta)^{-1}\big{|}_{L^2(\partial \Omega)}$. To show that $\Phi$ is linear contraction on $L^2(\partial \Omega)$ one applies the same sequence of inequalities in equation (\ref{contraction}). The associated extension $-i\hat{H}_\Phi$ is therefore maximally dissipative with domain 
\begin{equation}
        D(\hat{H}_\Phi)= \left\{\psi \in D(\hat{H}^*): -\tilde{\Theta}_\beta (i- \tilde{\Theta}_\beta)^{-1}\iota_-\tilde{\tau}_D \psi=(i - \tilde{\Theta}_\beta)^{-1}\iota_+ \tilde{\tau}_N \psi_D \right\}.
    \end{equation}
Repeating the same steps in the proof of Theorem \ref{Robin boundary conditions ugly} starting from equation (\ref{H Phi bad def}) returns $D(\hat{H}_\Phi)=  \left\{\psi \in H^1(\Omega)\cap D(\hat{H}^*):\tilde\tau_N \psi=i \beta \tilde\tau_D \psi \right\}$.
\end{proof}
\begin{remark}\label{Hardy remark}
    Condition \em (R2) \em allows the Robin boundary condition to include tangential derivatives along the boundary, while conditions \em(R1) \em and \em (R3) \em allow for singular multiplication operators. For example, if $\Omega\subset\mathbb{R}^n$ is a bounded $C^2$ domain of dimension $n> 2$ and $y \in \partial \Omega$, then the fractional Hardy inequality \cite{Hardy} implies that multiplication by $\frac{1}{|x-y|^{1/2}}$ defines a bounded linear operator from $H^{1/2}(\partial \Omega)\to L^2(\partial \Omega)$. Duality then shows that multiplication by $\frac{1}{|x-y|^{1/2}}$ extends to a bounded linear operator from $L^2(\partial \Omega)\to H^{-1/2}(\partial \Omega)$, therefore multiplication by $\frac{1}{|x-y|}$ defines a bounded positive operator from $H^{1/2}(\partial \Omega)\to H^{-1/2}(\partial \Omega)$.
\end{remark}
We now restate and prove the third main result of this paper, the asymptotic stability of all solutions to the initial-boundary value problem when the real part of $\beta$ is \textit{strictly positive}; meaning $\text{Re}\langle \beta \chi, \chi \rangle_{H^{-1/2}(\partial \Omega)\times H^{1/2}(\partial \Omega)}>0$ for all $0 \neq \chi \in H^{1/2}(\partial \Omega)$.
\begin{theoremmainagain}\label{Robin asymt stab again} Let $\Omega \subset \mathbb{R}^n$ be a bounded $C^2$ domain of dimension $n>1$, let $\hat{H}:=(-\Delta+V)\big{|}_{C_c^\infty(\Omega)}$ with $V \in L^\infty(\Omega,\mathbb{R})$, and suppose that $\beta$ satisfies the assumptions of Theorem \ref{Robin boundary conditions ugliest} while having strictly positive real part. Then all solutions of the initial-boundary value problem (\ref{Ugliest Robin IBVP}) asymptotically vanish, $||W_t\psi_0||_{L^2(\Omega)}\xrightarrow{t \to \infty}0$ for all $\psi_0 \in L^2(\Omega)$.
\end{theoremmainagain}
\begin{proof}
    We again follow the last steps in the proofs of Theorem \ref{Robin boundary conditions nice} and Theorem \ref{Robin boundary conditions ugly}. The linear contraction $\Phi=(i+\tilde{\Theta}_\beta)(i-\tilde{\Theta}_\beta)^{-1}\big{|}_{L^2(\partial \Omega)}$ is \textit{strictly} contractive on $L^2(\partial \Omega)$ by the same set of inequalities as in equation (\ref{contraction}). Once again $-i\hat{H}_\Phi$ has compact resolvent since for $\lambda \in \rho(-i\hat{H}_\Phi)$ the resolvent operator $(-i\hat{H}_\Phi-\lambda)^{-1}$ maps $L^2( \Omega)$ into $H^1( \Omega)$, which by Rellich-Kondrachov \cite{Brezis} embeds compactly into $L^2( \Omega)$. By the compactness of the resolvent and Lemma \ref{no eigenvalues}, Theorem \ref{asympt stability} states $||W_t\psi_0||_{L^2( \Omega)}\xrightarrow{t \to \infty}0$ for all $\psi_0 \in L^2(\Omega)$.
\end{proof}
\section{Detection Time Distributions}
Suppose that a quantum particle is initially prepared with state $\psi_0$ of unit norm in a Hilbert space $\mathcal{H}$, and suppose that the particle undergoes an idealized detection process which is irreversible, autonomous, and highly sensitive to interactions; so that the formation of any entanglement between the particle and the detector leads to an immediate firing. Then; as argued in the introduction; the dynamics of the particle is governed by a $C_0$ contraction semigroup $W_t=\exp(-itL)$ on $\mathcal{H}$ while the detector has not yet fired, with the quantity $||W_t\psi_0||_{\mathcal{H}}^2$ representing the probability that the particle remains undetected up to time $t$. Hence, the probability that the particle will be detected between times $0\leq t_1<t_2<\infty$ is given by
\begin{equation}\label{Prob dist}
    \text{Prob}_{\psi_0}(t_1<t<t_2):=||W_{t_1}\psi_0||_{\mathcal{H}}^2 - ||W_{t_2}\psi_0||_{\mathcal{H}}^2
\end{equation}
while the event that the particle is never detected; denoted by $t=\text{N}$; has probability
\begin{equation}\label{Prob no det}
    \text{Prob}_{\psi_0}(t=\text{N}):=\lim_{t \to \infty}||W_t\psi_0||_{\mathcal{H}}^2
\end{equation}
where this limit exists by monotonicity. Werner \cite{Werner1987}  demonstrated that every $C_0$ contraction semigroup $W_t$ on a Hilbert space $\mathcal{H}$ admits a ``Born rule" for the detection time distribution above. Formally stated, there exists a positive operator-valued measure defined on Lebesgue measureable subsets $I \subset [0,\infty)\cup \{\text{N}\}$ and acting on $\mathcal{H}$ such that
\begin{equation}
    \langle E(I)\psi_0,\psi_0\rangle_{\mathcal{H}}=\text{Prob}_{\psi_0}(t \in I)
\end{equation}
agrees with (\ref{Prob dist}) and (\ref{Prob no det}) for all $\psi_0 \in \mathcal{H}$. Werner begins his analysis by pointing out that
\begin{equation}
    \mathcal{J}(\psi,\phi):=-\frac{d}{dt}\langle W_t\psi, W_t\phi\rangle_{\mathcal{H}}\big{|}_{t=0}=\langle iL \psi,\phi\rangle_{\mathcal{H}}+\langle \psi,iL\phi\rangle_{\mathcal{H}}
\end{equation}
defines a non-negative Hermitian form on the domain of the semigroup generator $D(L)$. In most cases this Hermitian form may not be positive, so $D(L)$ is not a pre-Hilbert space. Werner only considered cases where $D(L)$ is a pre-Hilbert space, however, we can apply the Cauchy-Schwarz inequality to show that the set of $\psi\in D(L)$ such that $\mathcal{J}(\psi,\psi)=0$ is a linear subspace of $D(L)$, in fact it is the radical of $\mathcal{J}$
\begin{equation}
    \text{Rad}(\mathcal{J}):=\{\psi \in D(L): \mathcal{J}(\psi,\phi)=0 \quad \forall \phi \in D(L) \}=\{\psi \in D(L): \mathcal{J}(\psi,\psi) =0\}.
\end{equation}
Taking the quotient of $D(L)$ with $\text{Rad}(\mathcal{J})$ returns a pre-Hilbert space $(D(L)/\text{Rad}(\mathcal{J}), \mathcal{J}(\cdot,\cdot))$. Let $\mathcal{K}$ denote the completion of this pre-Hilbert space, and let $j:D(L)\to \mathcal{K}$ denote the canonical embedding, so the inner product on $\mathcal{K}$ is defined on the dense set $j(D(L)) \subset \mathcal{K}$
\begin{equation}\label{Embedding In Prod}
    \langle j \psi,j\phi \rangle_{\mathcal{K}}=\langle iL \psi, \phi \rangle_{\mathcal{H}}+\langle \psi, iL \phi \rangle_{\mathcal{H}}.
\end{equation}
For each $\psi \in D(L)$ we consider the function $J\psi:[0,\infty) \to \mathcal{K}$ defined by $(J\psi)(t):=j(W_t\psi)$. Then
\begin{equation}
    \int_{0}^\infty ||(J \psi)(t) ||_{\mathcal{K}}^2 ~ dt=- \int_0^\infty \frac{d}{dt}||W_t \psi||_\mathcal{H}^2 ~dt=||\psi||_\mathcal{H}^2 - \lim_{t \to \infty}||W_t \psi||_{\mathcal{H}}^2 \leq ||\psi||_{\mathcal{H}}^2.
\end{equation}
where these limits exist by the monotonicity of $||W_t\psi||_{\mathcal{H}}^2$. This bound shows that $J$ extends from $D(L)$ to a linear contraction $J:\mathcal{H}\to L^2([0,\infty),\mathcal{K})$. The quantity $||(J\psi)(t)||_\mathcal{K}^2$ measures the rate at which probability flows from the space of states where the detector is primed $\mathcal{H}_p$ to the space of states where the detector has fired $\mathcal{H}_F$. It therefore acts as a probability density for the particle's time of detection. Now, for any Lebesgue measureable subset $I\subset [0,\infty)$, we introduce the operator $E(I):=J^*\chi_IJ$ acting on $\mathcal{H}$, with $\chi_I$ denoting the indicator function. Then
\begin{equation}
    \langle E(I)\psi_0,\psi_0\rangle_{\mathcal{H}}=\int_{I} ||(J\psi)(t)||_{\mathcal{K}}^2 ~dt \geq 0,
\end{equation}
so $E(I)$ is a positive operator. It is straightforward to show that $E(\cdot)$ forms a (unnormalized) positive operator-valued measure, which, when extended by setting $E(\{N\})=1-J^*J$ forms a normalized positive operator-valued measure on $[0,\infty)\cup \{\text{N}\}$ that agrees with (\ref{Prob dist}) and (\ref{Prob no det}. 

Any Hilbert space $\mathcal{K}$ with a linear map $j:D(L) \to \mathcal{K}$ satisfying equation (\ref{Embedding In Prod}) is called an \textit{exit space}. All $C_0$ contraction semigroups admit at least one exit space; such as the abstract \textit{minimal exit space} detailed above; and the detection time probability density $||(J\psi)(t)||_{\mathcal{K}}^2$ and positive operator-valued measure $E(\cdot)$ is independent of the choice of exit space. The last main result of this paper applies the theory of boundary quadruples to explicitly construct exit spaces for $C_0$ contraction semigroups $W_t$ that ``weakly solve" a Schr\"odinger equation $i\partial_t \psi=\hat{H}^*\psi$ for some densely defined symmetric operator $\hat{H}$ on $\mathcal{H}$.
\begin{theoremmainagain}
    Let $-i\hat{H}$ be a densely defined skew-symmetric operator on a Hilbert space $\mathcal{H}$ and let $W_t$ be a $C_0$ contraction semigroup on $\mathcal{H}$ whose generator is extended by $-i\hat{H}^*$. Then for any choice of boundary quadruple $(\mathcal{H}_\pm,G_\pm)$ for $-i\hat{H}$ we may write $W_t=\exp(-it\hat{H}_\Phi)$ where $\Phi:\mathcal{H}_-\to\mathcal{H}_+$ is the linear contraction guaranteed by Theorem \ref{Description}, and \begin{equation}
        j:D(\hat{H}_\Phi)\to \mathcal{H}_-, \quad j:\psi \mapsto \sqrt{1-\Phi^*\Phi}G_-\psi
    \end{equation} 
    defines an exit space for $W_t$. Consequently $(J\psi_0)(t):=\sqrt{1-\Phi^*\Phi}G_-W_t\psi_0$ extends to a linear contraction $J:\mathcal{H}\to L^2([0,\infty),\mathcal{H}_-)$ with $||\sqrt{1-\Phi^*\Phi}G_-W_t\psi_0||_{\mathcal{H}_-}^2$ equal to the (unnormalized) probability density for the time of detection. Lastly, the positive operator-valued measure for the detection time probability distribution can be expressed as $E(I)=J^*\chi_IJ$ for Lebesgue measureable subsets $I\subset [0,\infty)$ and $E(\{N\})=\mathds{1}_\mathcal{H}-J^*J$.
\end{theoremmainagain}

\begin{proof}
    Let $\Phi^*:H_+ \to H_-$ denote the Hilbert space adjoint of $\Phi$. Then for all $\psi, \phi \in D(\hat{H}_\Phi)$ 
    \begin{equation}
    \begin{split}
        \langle i\hat{H}_\Phi \psi, \phi \rangle_{\mathcal{H}}+\langle \psi, i\hat{H}_\Phi \phi \rangle_{\mathcal{H}} &= \langle G_-\psi, G_- \phi \rangle_{H_-}- \langle G_+ \psi, G_+\phi \rangle_{H_+}
        \\
        &= \langle G_-\psi, G_- \phi \rangle_{H_-}- \langle \Phi G_- \psi, \Phi G_-\phi \rangle_{H_+}
        \\
        &= \langle (1- \Phi^* \Phi)G_-\psi, G_- \phi \rangle_{H_-}
        \\
        &= \langle \sqrt{1- \Phi^* \Phi}G_-\psi, \sqrt{1-\Phi^* \Phi} G_- \phi \rangle_{H_-}
    \end{split}
    \end{equation}
    where the square root of the bounded positive operator $1-\Phi^*\Phi$ on $\mathcal{H}_-$ is defined by a power series.
\end{proof}
 \section{Summary and Outlook}
 We have shown that wave functions of non-relativistic quantum particles undergoing idealized irreversible hard autonomous detection along the boundary of some bounded $C^2$ region $\Omega\subset \mathbb{R}^n$ must evolve according to the Schr\"odinger equation while satisfying a time-independent absorbing boundary condition along $\partial \Omega$. Every such detector model admits a natural Born rule for the distribution of times at which the particle is detected along $\partial \Omega$, and we have shown that a detection will almost surely occur in finite time whenever the particle is completely surrounded by screen detectors.  
 \par
 There are several exciting directions to go in extending these results. The local Robin boundary conditions studied in this paper offer a large family of possible hard detector models, but it is worthwhile to also look into other types of local absorbing boundary conditions. The modern theory of boundary triples for Schr\"odinger operators \cite{BEHRNDT2014,BEHRNDT20155903} might allow us to extend the theorems in this paper to the case of unbounded Lipschitz domains $\Omega \subset \mathbb{R}^n$. A boundary triple construction for the Dirac Hamiltonian is nearly complete \cite{SADIRAC, DiracSing} and it would be of great interest to prove similar results for spin $\frac{1}{2}$ relativistic quantum particles undergoing irreversible hard autonomous detection. We also plan to investigate in an upcoming paper \cite{DetTime} the behavior of the detection time probability distribution under perturbations of the boundary condition. 
 
 Another avenue of research might focus on providing similar parameterization results when one of our idealizations regarding the detecting screen is relaxed. If one allows the detection mechanism to vary with time, then the dynamics of the wave function are not autonomous, and one would investigate whether the results of Wegner \cite{Wegner2017} and Arendt et al. \cite{Arendt2023} generalize to the setting where $W_t$ does not form a semigroup under composition. If one instead allows $\Psi_t\big{|}_{\mathcal{H}_P}$ to become entangled in its dynamics, then it would be interesting to see if a similar parameterization result holds for the particle's density matrix dynamics. 
\section{Acknowledgment}
The author thanks Shadi Tahvildar-Zadeh, Markus Holzmann, Jussi Behrndt, Siddhant Das, Will Cavendish, Anupam Nayak, Mark Vaysiberg, Sheldon Goldstein, and especially Roderich Tumulka  for the insight and advice they provided in the development of this manuscript.
\appendix
\section{Dirichlet-to-Neumann Maps}\label{D to N section}
This appendix is dedicated to the study of Dirichlet boundary value problems and the Dirichlet-to-Neumann map. We assume throughout this section that $\Omega \subset \mathbb{R}^n$ is a bounded $C^2$ domain of dimension $n>1$, $V \in L^\infty(\Omega,\mathbb{R})$, and we take $\hat{H}:=-\Delta+V$ densely defined on the minimal domain $D(\hat{H})=H^2_0(\Omega)$. Recall that the Dirichlet extension $\hat{H}_D$ of $\hat{H}$ defined by
\begin{equation}
    D(\hat{H}_D):=\{ \psi \in H^2(\Omega): \tau_D \psi=0\}, \quad \hat{H}_D\psi:=(-\Delta +V)\psi
\end{equation}
is a self-adjoint operator on $L^2(\Omega)$. This operator can be used to prove an existence and uniqueness result for a class of Dirichlet boundary value problems.
\begin{lemma}\label{Dirichlet solution properties}
    Let $\lambda \in \rho(\hat{H}_D)\cap \mathbb{R}$ and $s \in [-\frac{1}{2},\frac{3}{2}]$. For each $\xi \in H^{s}(\partial \Omega)$, there exists a unique $\psi_\lambda \in H^{s+1/2}(\Omega)$ to the boundary value problem
    \begin{equation}\label{Boundary value problem}
        \begin{cases}
            (\hat{H}^*-\lambda)\psi=0
            \\
            \psi \big{|}_{\partial \Omega}=\xi
        \end{cases}
    \end{equation}
    The solution mappings $\xi \mapsto\psi_\lambda$ define a bounded linear map $\gamma(\lambda):H^s(\partial \Omega)\to H^{s+1/2}(\Omega)$.
    \end{lemma}
\begin{proof}
 We first prove the existence and uniqueness of a solution to (\ref{Boundary value problem}) for $\xi \in H^{-1/2}(\partial \Omega)$. By Lemma \ref{Trace Extensions} there exists some $\psi \in D(\hat{H}^*)$ such that $\tilde{\tau}_D\psi=\xi$. It is then easy to verify that $\psi_\lambda:=\psi-(\hat{H}_D-\lambda)^{-1}(\hat{H}^*-\lambda)\psi$ resides in $\ker(\hat{H}^*-\lambda)$ and satisfies $\tilde{\tau}_D\psi_\lambda=\tilde{\tau}_D\psi=\xi$, and is hence a solution to (\ref{Boundary value problem}). For uniqueness, we observe that if $\psi_\lambda$ and $\psi_\lambda'$ both solve (\ref{Boundary value problem}) then $\psi_\lambda-\psi_\lambda'\in \ker(\tau_D)\cap \ker(\hat{H}^*-\lambda)=\{0\}$, where the last equality follows from $\lambda \in \rho(\hat{H}_D)$. Hence, the solution mapping $\xi \mapsto \gamma(\lambda)\xi:=\psi_\lambda$ is a well-defined linear map from $H^{-1/2}(\partial \Omega)\to L^2(\Omega)$. In fact, if $\xi \in H^{3/2}(\partial \Omega)$ then Lemma \ref{Trace Definition} ensures that there exists a $\psi \in H^2(\Omega)$ such that $\tau_D\psi=\xi$, so $\psi_\lambda=\psi - (\hat{H}_D-\lambda)^{-1}(\hat{H}^* - \lambda)\psi$ is a difference of two $H^2(\Omega)$ functions and is thus in $H^2(\Omega)$. So $\gamma(\lambda)$ maps $H^{3/2}(\partial \Omega)\to H^2(\Omega)$. This establishes existence and uniqueness of solutions at the end-point cases $s=-\frac{1}{2}$ and $s=\frac{3}{2}$.
        
For continuity at the end-point cases we will show that $\gamma(\lambda)$ is the Banach space adjoint of $-\tau_N (\hat{H}_D-\lambda)^{-1}:L^2(\Omega)\to H^{1/2}(\partial \Omega)$, making it a bounded operator from $H^{-1/2}(\partial \Omega)\to L^2(\Omega)$. To see this, let $\xi \in H^{-1/2}(\partial \Omega)$ and $\phi \in L^2(\Omega)$. Then $\gamma(\lambda)\xi \in \ker(\hat{H}^*-\lambda)$ and we have
\begin{equation}
    \begin{split}
    \langle \gamma(\lambda)\xi, \phi \rangle_{L^2(\Omega)}&=\langle \gamma(\lambda)\xi, (\hat{H}^* - \lambda)(\hat{H}_D - \lambda)^{-1}\phi \rangle_{L^2(\Omega)}
        \\
    &=\langle \gamma(\lambda)\xi, \hat{H}^*(\hat{H}_D - \lambda)^{-1}\phi \rangle_{L^2(\Omega)}- \langle \lambda \gamma(\lambda)\xi, (\hat{H}_D-\lambda)^{-1}\phi \rangle_{L^2(\Omega)}
        \\
    &=\langle \gamma(\lambda)\xi, \hat{H}^*(\hat{H}_D - \lambda)^{-1}\phi \rangle_{L^2(\Omega)}- \langle \hat{H}^* \gamma(\lambda)\xi, (\hat{H}_D-\lambda)^{-1}\phi \rangle_{L^2(\Omega)}
        \\
    &=-\langle \xi, \tau_N (\hat{H}_D - \lambda)^{-1}\phi \rangle_{H^{-1/2}(\partial \Omega)\times H^{1/2}(\partial \Omega)},
 \end{split}
\end{equation}
  where in the last step we applied the abstract Green's formula (\ref{Green's Identity Star}) along with $\tilde{\tau}_D\gamma(\lambda)\xi=\xi$ and $\tau_D (\hat{H}_D-\lambda)^{-1}\phi=0$. Hence $\gamma(\lambda)$ is a bounded operator from $H^{-1/2}(\partial \Omega)\to L^2(\Omega)$ that also maps $H^{3/2}(\partial \Omega)\to H^2(\Omega)$. The closed graph theorem then implies that the restriction $\gamma(\lambda):H^{3/2}(\partial \Omega)\to H^2(\Omega)$ is also a bounded operator, and interpolation (see e.g. \cite[Theorem 1.4.3.3]{Grisvard} and \cite[Theorems 5.1 and 7.7]{LH1972}) returns that the restrictions $\gamma(\lambda):H^s(\partial \Omega)\to H^{s+1/2}(\Omega)$ are bounded for $s \in [-\frac{1}{2},\frac{3}{2}]$.
  \end{proof}
  \begin{corollary}\label{Converse Trace Regularity}
	If $\psi \in D(\hat{H}^*)$ has trace $\tilde{\tau}_D \psi \in H^s(\partial \Omega)$ for some $s \in [-\frac{1}{2},\frac{3}{2}]$, then $\psi \in H^{s+1/2}(\Omega)$.
	\end{corollary}
    \begin{proof}
        Let $\psi \in D(\hat{H}^*)$ with $\tilde{\tau}_D\psi \in H^s(\partial \Omega)$, and let $\lambda \in \rho(\hat{H}_D)\cap \mathbb{R}$. Then $\psi-\gamma(\lambda)\tilde{\tau}_D\psi \in \ker(\tilde{\tau}_D)\subset H^2(\Omega)$, so $\psi$ is the sum of two elements in $H^{s+1/2}(\Omega)$ and is thus in $H^{s+1/2}(\Omega)$.
    \end{proof}
    We now turn our attention to the ``Dirichlet-to-Neumann" map. \begin{definition}\label{Dichlet to Neumann}
 For $\lambda \in \rho(\hat{H}_D)\cap \mathbb{R}$, the ``Dirichlet-to-Neumann" map is defined by
    \begin{equation}\label{Dirichlet to Neumann Eq again}
        D(\lambda):H^{3/2}(\partial \Omega)\to H^{1/2}(\partial \Omega), \quad D(\lambda)\xi:=\tau_N \gamma(\lambda)\xi.
    \end{equation}
\end{definition}
Clearly $D(\lambda)$ is a bounded linear operator since it is a composition of two bounded linear operators, and $D(\lambda)$ can be continuously extended to a bounded operator $\tilde{D}(\lambda):=\tilde{\tau}_N\gamma(\lambda):H^{-1/2}(\partial \Omega)\to H^{-3/2}(\partial \Omega)$ using the extension $\tilde{\tau}_N$. We now prove several useful facts about the extended ``Dirichlet-to-Neumann" map. We should recall that the Neumann extension $\hat{H}_N$ of $\hat{H}$ defined by 
\begin{equation}
    D(\hat{H}_N):=\{ \psi \in H^2(\Omega): \tau_N \psi=0\}, \quad \hat{H}_N\psi:=(-\Delta +V)\psi
\end{equation}
is a self-adjoint operator on $L^2(\Omega)$.
\begin{lemma}\label{Dirichlet to Neumann Facts}
    For $\lambda \in \rho(\hat{H}_D)\cap \mathbb{R}$, the extended ``Dirichlet-to-Neumann" map $\tilde{D}(\lambda):H^{-1/2}(\partial \Omega)\to H^{-3/2}(\partial \Omega)$ is equal to the Banach space adjoint of $D(\lambda):H^{3/2}(\partial \Omega)\to H^{1/2}(\partial \Omega)$. In addition, its restrictions $\tilde{D}(\lambda):H^s(\partial \Omega)\to H^{s-1}(\partial \Omega)$ are bounded linear operators for $s \in [-\frac{1}{2},\frac{3}{2}]$, with $\tilde{D}(\lambda):H^{1/2}(\partial \Omega)\to H^{-1/2}(\partial \Omega) $ ``symmetric" in the sense that $\langle \tilde{D}(\lambda)\zeta,\zeta \rangle_{H^{-1/2}(\partial \Omega)\times H^{1/2}(\partial \Omega)}\in \mathbb{R}$ for all $\zeta \in H^{1/2}(\partial \Omega)$. For $\lambda \in \rho(\hat{H}_D)\cap \rho(\hat{H}_N)\cap \mathbb{R}$ the operators $\tilde{D}(\lambda):H^s(\partial \Omega)\to H^{s-1}(\partial \Omega)$ are bijective linear maps with bounded inverse for $s \in [-\frac{1}{2},\frac{3}{2}]$. 
\end{lemma}
\begin{proof}
    Let $\lambda \in \rho(\hat{H}_D)\cap \mathbb{R}$, $\xi \in H^{-1/2}(\partial \Omega)$, and $\chi \in H^{3/2}(\partial \Omega)$. Setting $\psi_\lambda=\gamma(\lambda)\xi\in \ker(\hat{H}^*-\lambda)$ and $\phi_\lambda=\gamma(\lambda)\chi\in H^2(\Omega)\cap \ker(\hat{H}^*-\lambda)$, we can compute from the extended Green's identity (\ref{Green's Identity Star}) that
    \begin{equation}\label{D is symmetric}
    \begin{split}
        \langle \tilde D (\lambda)\xi,\chi\rangle_{H^{-3/2}(\partial \Omega)\times H^{3/2}(\partial \Omega)}&=\langle \tilde \tau_N \psi_\lambda,\tau_D \phi_\lambda\rangle_{H^{-3/2}(\partial \Omega)\times H^{3/2}(\partial \Omega)}
        \\
        &=\langle \tilde \tau_D \psi_\lambda, \tau_N \phi_\lambda \rangle_{H^{-1/2}(\partial \Omega)\times H^{1/2}(\partial \Omega)}
        \\
        &=\langle \xi,D(\lambda)\chi\rangle_{H^{-1/2}(\partial \Omega)\times H^{1/2}(\partial \Omega)}
    \end{split}
    \end{equation}
    where in the second line we used $\langle -i\hat{H}^* \psi_\lambda, \phi_\lambda \rangle_{L^2(\partial \Omega)}+\langle  \psi_\lambda, -i\hat{H}^*\phi_\lambda \rangle_{L^2(\partial \Omega)}=0$. Hence $\tilde{D}(\lambda)$ is equal to the Banach space adjoint of $D(\lambda)$. Now $\tilde{D}(\lambda)$ is a bounded linear operator from $H^{-1/2}(\partial \Omega)\to H^{-3/2}(\partial \Omega)$ which also maps $H^{3/2}(\partial \Omega)\to H^{1/2}(\partial \Omega)$, so interpolation (see e.g. \cite[Theorem 1.4.3.3]{Grisvard} and \cite[Theorems 5.1 and 7.7]{LH1972}) implies the restrictions
    \begin{equation}
        \tilde{D}(\lambda):H^{s}(\partial \Omega)\to H^{s-1}(\partial \Omega)
    \end{equation}
    are bounded for $s\in [-\frac{1}{2},\frac{3}{2}]$. In particular, $\tilde{D}(\lambda):H^{1/2}(\partial \Omega)\to H^{-1/2}(\partial \Omega)$ is ``symmetric" by equation (\ref{D is symmetric}).
    \par
    Now let $\lambda \in \rho(\hat{H}_D)\cap \rho(\hat{H}_N)\cap \mathbb{R}$. We wish to prove the bijectivity of the extended operator $\tilde{D}(\lambda):H^{-1/2}(\partial \Omega)\to H^{-3/2}(\partial \Omega)$. To show injectivity, suppose $\tilde{D}(\lambda)\xi=0$ for some $\xi \in H^{-1/2}(\partial \Omega)$. Then $\gamma(\lambda)\xi \in D(\hat{H}_N)\cap\ker(\hat{H}^*-\lambda)=\{0\}$, hence $\xi=\tilde{\tau}_D0=0$. To prove surjectivity, let $\chi \in H^{-3/2}(\partial \Omega)$. By Lemma \ref{Trace Extensions} there exists some $\psi \in D(\hat{H}^*)$ such that $\tilde{\tau}_N\psi=\chi$. Setting $\psi_\lambda=\psi-(\hat{H}_N - \lambda)^{-1}(\hat{H}^*-\lambda)\psi\in \ker(\hat{H}^*-\lambda)$, we have $\tilde{\tau}_N\psi_\lambda=\tilde{\tau}_N\psi=\chi$. So, $\xi=\tilde{\tau}_D\psi_\lambda \in H^{-1/2}(\partial \Omega)$ satisfies $\tilde{D}(\lambda)\xi=\tilde{\tau}_N \psi_\lambda=\chi$ as desired. It follows that $\tilde{D}(\lambda):H^{-1/2}(\partial \Omega)\to H^{-3/2}(\partial \Omega)$ is bijective with $\tilde{D}^{-1}(\lambda):H^{s-1}(\partial \Omega)\to H^s(\partial \Omega)$ a bounded linear operator for each $s \in [-\frac{1}{2}, \frac{3}{2}]$.
\end{proof}

\end{document}